\documentclass[acmsmall]{acmart}

\setcopyright{rightsretained}
\acmDOI{10.1145/3656418}
\acmYear{2024}
\copyrightyear{2024}
\acmSubmissionID{pldi24main-p263-p}
\acmJournal{PACMPL}
\acmVolume{8}
\acmNumber{PLDI}
\acmArticle{188}
\acmMonth{6}
\received{2023-11-16}
\received[accepted]{2024-03-31}

\usepackage{booktabs}   
\usepackage{subcaption} 
\usepackage{algorithm}
\usepackage[noend]{algpseudocode}
\usepackage{amsmath}
\usepackage{stmaryrd}
\usepackage{url}
\usepackage{listings}
\usepackage[bottom]{footmisc}
\usepackage{graphicx}
\usepackage{wrapfig}
\usepackage{fancyvrb}
\usepackage{paralist}
\usepackage{caption}
\usepackage{comment}
\usepackage{xspace}
\usepackage{multirow}
\usepackage{booktabs}
\usepackage{enumerate}
\usepackage{enumitem}
\usepackage{caption}
\usepackage{graphicx}
\usepackage{listing}
\usepackage{tikz}
\usepackage{ifsym}
\usepackage{circledsteps}
\usepackage[frozencache,cachedir=minted-cache]{minted}
\usepackage{xcolor}
\usepackage{soul}
\usepackage{appendix}
\usepackage[nameinlink]{cleveref}
\usepackage{textgreek}

\usepackage[utf8]{inputenc}
\usepackage{pgfplots}
\DeclareUnicodeCharacter{2212}{−}
\usepgfplotslibrary{groupplots,dateplot}
\usetikzlibrary{patterns,shapes.arrows}
\pgfplotsset{compat=newest}

\newtheorem{lemma}{Lemma}

\crefname{problem}{problem}{problems}
\crefname{lemma}{lemma}{lemmas}

\makeatletter
\newcommand{\pequiv}[3]{%
    \ifinalign@
        #1 &\equiv \left(#2, #3\right)
    \else
        #1 \equiv \left(#2, #3\right)
    \fi
}
\newcommand{\requiv}[3]{%
    \ifinalign@
        {#1} &\equiv_{#2} {#3}
    \else
        {#1} \equiv_{#2} {#3}
    \fi
}

\newcommand{\idiotbox}[2]{
\fbox{
\begin{minipage}{0.9\linewidth}
{\bf Result for {#1}:}
{#2}
\end{minipage}
}
}

\newcommand{\irule}[2]%
   {\mkern-2mu\displaystyle\frac{#1}{\vphantom{,}#2}\mkern-2mu}
\newcommand{\irulelabel}[3]
{
\mkern-2mu
\begin{array}{ll}
\displaystyle\frac{#1}{\vphantom{,}#2} & \!\!\!\! #3
\end{array}
\mkern-2mu
}

\newcommand{\bfpara}[1]{\vspace{5pt}\noindent\emph{\textbf{#1}}}

\newcommand{\set}[1]{\{ #1 \}}
\newcommand{\listappend}{+\!+}
\newcommand{\fst}{\textsf{fst}}
\newcommand{\foldl}{\textsf{foldl}}
\newcommand{\map}{\textsf{map}}
\newcommand{\filter}{\textsf{filter}}
\newcommand{\length}{\textsf{length}}

\newcommand{\ite}[3]{#1~?~#2:#3}
\newcommand{\dom}{\textsf{dom}}

\newcommand{\nil}{\mathsf{Nil}}
\newcommand{\last}{\textsf{last}}
\newcommand{\hoare}[3]{\set{#1}~#2~\set{#3}}
\newcommand{\emptymap}{\set{\ }}

\newcommand{\denot}[1]{\llbracket #1 \rrbracket}

\newcommand{\prog}{\mathcal{P}}
\newcommand{\init}{\mathcal{I}}

\newcommand{\scheme}{\mathcal{S}}
\newcommand{\rfs}{\Phi}
\newcommand{\sketch}{\Omega}
\newcommand{\context}{\Delta}
\newcommand{\rewrite}{\hookrightarrow}
\newcommand{\hole}{\Box}
\newcommand{\pp}{\nolinebreak{\bf +}\nolinebreak\hspace{-.10em}{\bf +}}
\newcommand{\axioms}{\mathcal{A}}

\newcommand{\tool}{\textsc{Opera}\xspace}
\newcommand{\toolname}{\tool}

\lstdefinestyle{numbers}
{
  numbers=left,
  numberstyle=\sf,
  xleftmargin=15pt
}

\newmintinline[T]{haskell}{fontsize=\small}
\newmintinline[Tfn]{haskell}{fontsize=\footnotesize, escapeinside=||}

\begin{document}

\title{From Batch to Stream: Automatic Generation of Online Algorithms}


\author{Ziteng Wang}
\orcid{0009-0001-8487-8093}
\affiliation{%
  \institution{University of Texas at Austin}
  \city{Austin, TX}
  \country{USA}
}
\email{ziteng@utexas.edu}

\author{Shankara Pailoor}
\orcid{0000-0002-9253-9585}
\affiliation{%
  \institution{University of Texas at Austin}
  \city{Austin, TX}
  \country{USA}
}
\email{spailoor@cs.utexas.edu}

\author{Aaryan Prakash}
\orcid{0009-0004-3142-3345}
\affiliation{%
  \institution{University of Texas at Austin}
  \city{Austin, TX}
  \country{USA}
}
\email{aaryanprakash@utexas.edu}

\author{Yuepeng Wang}
\orcid{0000-0003-3370-2431}
\affiliation{%
  \institution{Simon Fraser University}
  \city{Burnaby}
  \country{Canada}
}
\email{yuepeng@sfu.ca}

\author{Işıl Dillig}
\orcid{0000-0001-8006-1230}
\affiliation{%
  \institution{University of Texas at Austin}
  \city{Austin}
  \country{USA}
}
\email{isil@cs.utexas.edu}



\begin{abstract}
Online streaming algorithms, tailored for continuous data processing, offer substantial benefits but are often more intricate to design than their offline counterparts. This paper introduces a novel approach for automatically synthesizing online streaming algorithms from their offline versions. In particular, we propose a novel methodology, based on the notion of \emph{relational function signature (RFS)}, for deriving an online algorithm given its offline version. Then, we propose a concrete synthesis algorithm that is an instantiation of the proposed methodology. Our algorithm uses the RFS to decompose the synthesis problem into a set of independent subtasks and uses a combination of symbolic reasoning and search to solve each subproblem. We implement the proposed technique in a new tool called \tool and evaluate it on over 50 tasks spanning two domains: statistical computations and online auctions. Our results show that \tool can automatically derive the online version of the original algorithm for  98\% of the tasks. Our experiments also demonstrate that \tool significantly outperforms alternative approaches, including adaptations of SyGuS solvers to this problem as well as two of \tool's own ablations.
\end{abstract}

\maketitle

\section{Introduction}\label{sec:intro}
The increasing demand for analyzing large volumes of data has sparked considerable interest in stream processing frameworks like Apache Flink~\cite{flink}, Spark Streaming~\cite{spark-streaming}, Kafka~\cite{kafka}, and others~\cite{storm, samza}. Because streaming applications process data as it arrives in a continuous fashion, they  can derive significant advantages from using \emph{online streaming algorithms} (online algorithms for short). In contrast to \emph{offline algorithms} that receive the input data in a single batch, online algorithms are designed to process data incrementally, without requiring access to the entire data set at once. 

Despite the potential advantages of online algorithms in many scenarios, offline algorithms are often easier to design than their online counterparts~\cite{so_post1,so_post2,so_post3,so_post4,so_post5,sandia-report}.
As an example, Figure~\ref{fig:python-var-two-pass} shows the implementation of an offline algorithm for calculating statistical variance for a list of numbers. Its online version, on the other hand, is known as Welford's algorithm~\cite{welford} and, as shown in  Figure~\ref{fig:python-var-welford},  it is significantly more complex than its offline version. In particular, note that the online version takes as input several \emph{auxiliary parameters} ({\tt v}, {\tt s}, {\tt sq}, {\tt n}) and, in addition to returning the variance, the algorithm also needs to compute the updated values of these parameters.
\par
This paper proposes a new technique for automatically synthesizing online algorithms from their offline version.  At a high level, the problem addressed in this paper falls under the general umbrella of \emph{incremental computation} on which there is a significant body of work~\cite{annieliu, higherorder, ingress, rpai, linalg, ras_dp, sun_synthesizing_2023, adapton, nominaladapton, acar_2005, acar_2006, acar_2008, acar_2014}. However, as discussed in more detail later (see Section~\ref{sec:related}), most prior work in this space focuses on programming language support and runtime systems for incrementalization~\cite{acar_2005, acar_2006, acar_2008, acar_2014, adapton, nominaladapton}. There is also some prior research on \emph{generating} incremental algorithms, but existing techniques are  either domain-specific~\cite{ingress, rpai, linalg, ras_dp, sun_synthesizing_2023}, or require hand-crafted rewrite rules to derive the target program~\cite{annieliu, higherorder}. 

In contrast to existing techniques, we propose a \emph{fully automated} and \emph{general} method for synthesizing online algorithms.
Given an offline algorithm  $\prog$  over input list $xs$, our method can automatically generate its online implementation scheme $\scheme = (\init, \prog')$ consisting of an \emph{initializer} $\init$ and \emph{online algorithm} $\prog'$. Here, the initializer specifies the computation result  for an empty list, and $\prog'$  \emph{incrementally} computes the output  given \emph{only} the previous computation result and a new stream element. Our approach can automatically derive both the initializer and the online algorithm and ensures that the synthesized  scheme is semantically equivalent to its offline version.


From a technical perspective, this paper makes two key contributions. The first one is a new \emph{synthesis methodology} for deriving online schemes, and the second contribution is a \emph{concrete synthesis algorithm} that is an instantiation of this methodology. Our methodology hinges upon the concept of a \emph{relational function signature (RFS)} which relates parameters of the online algorithm to computation results in the offline version. At a high level, the RFS (which is inferred automatically) serves as a relational specification between the offline and online algorithms and drives the entire synthesis process. In particular, our methodology relies on the notion of \emph{inductiveness relative to} an RFS and can be shown to be both sound and (under certain realistic assumptions) complete. 

A second technical contribution of this paper is a new synthesis algorithm that is an instantiation of the proposed  methodology. As shown schematically in Figure~\ref{fig:overview}, our method first statically analyzes the offline program to infer a suitable relational function signature. Along with the source code of the offline program, this RFS is used to generate a program sketch of the online algorithm. Crucially, the RFS-guided synthesis methodology ensures that each unknown in this sketch can be solved  \emph{completely independently}. This yields several independent synthesis subtasks, each requiring the discovery of an expression that satisfies its specification \emph{modulo} the RFS. However, because these expressions can nevertheless be quite large,  existing synthesis techniques (e.g., based on enumerative search) struggle to synthesize such expressions that arise in realistic online algorithms.  Our  algorithm addresses this challenge by proposing a novel \emph{expression synthesis} technique that marries the power of symbolic reasoning  with the flexibility of search.


\begin{figure}[!t]
\vspace{-0.1in}
\includegraphics[scale=0.3]{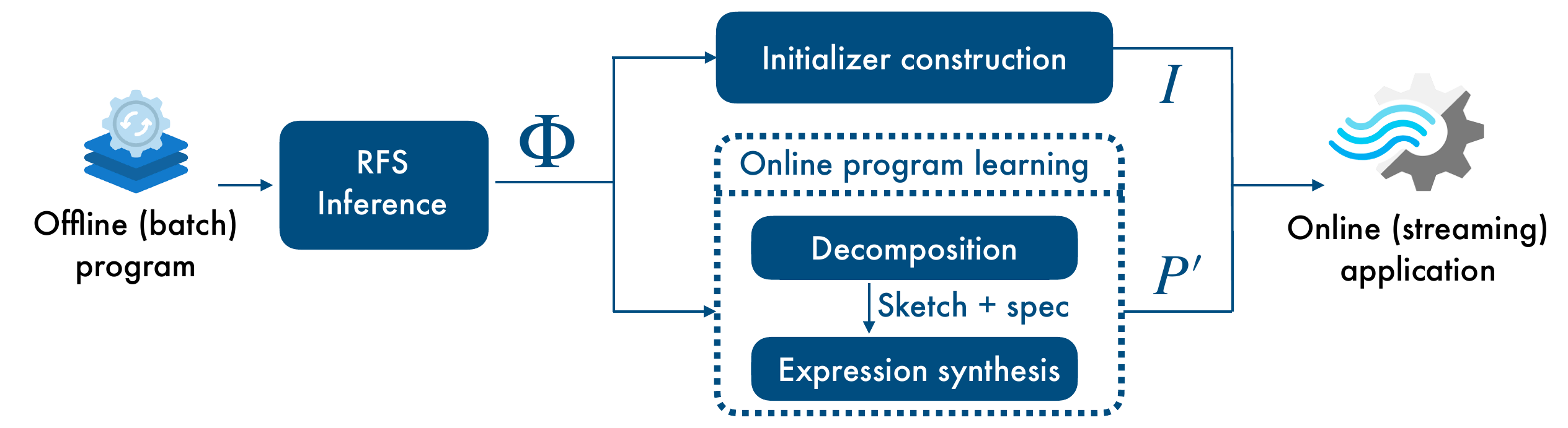}
\vspace{-0.1in}
\caption{Schematic illustration of our synthesis methodology}\label{fig:overview}
\vspace{-0.1in}
\end{figure}

{We have implemented our approach in a tool called \tool\footnote{\tool stands for Online Program gEneRAtor} and evaluated it on more than 50 offline-online conversion tasks spanning two domains (statistical algorithms and online auctions).  Our evaluation shows that \tool can automate this conversion process for $98\%$ of the tasks. We also perform a comparison against two baselines by adapting SyGuS solvers to this task and show that our proposed approach significantly outperforms these baselines: in particular, \tool can solve $2.6\times$ as many tasks as its closest competitor. Finally, we present ablation studies that quantify the relative importance of the different ingredients of our approach.}

To summarize, this paper makes the following key contributions:
\begin{itemize}[leftmargin=*]
\item We propose a novel synthesis methodology for deriving {online algorithms} from their offline counterparts which is based on the concept of \emph{relational function signatures} (RFS) and  the notion of \emph{inductiveness relative to an RFS}. 
\item We describe a concrete synthesis algorithm that is an instantiation of our proposed RFS-based methodology. This  algorithm  decomposes the overall synthesis problem into a set of completely {independent} sub-tasks and utilizes a novel expression synthesis technique that marries the power of symbolic reasoning with the flexibility of search. As shown experimentally, both of these ingredients are important for the practicality of our approach.
\item We implement our algorithm in a tool called \tool and use it to derive 
online versions of 51 offline algorithms. \tool can solve all benchmarks except one and significantly outperforms two baselines, which are problem-specific adaptations of leading SyGuS solvers.
\end{itemize}

\section{Overview}\label{sec:overview}

\begin{figure}[!t]
\centering
\begin{subfigure}[b]{0.45\linewidth}
\begin{minted}
[
fontsize=\footnotesize,
escapeinside=||,
numbersep=5pt,
linenos
]
{python}
def variance(xs):
    s = 0
    for x in xs:
        s += x
    avg = s / len(xs)

    sq = 0
    for x in xs:
        sq += (x - avg) ** 2
    return sq / len(xs)
\end{minted}
\caption{Two-pass algorithm for offline variance.}
\label{fig:python-var-two-pass}
\end{subfigure}
~\qquad\qquad~
\begin{subfigure}[b]{0.45\linewidth}
\begin{minted}
[
fontsize=\footnotesize,
escapeinside=||,
numbersep=5pt,
linenos
]
{python}
def welford(v, s, sq, n, x):
    new_s = s + x
    new_n = n + 1
    avg = new_s / new_n
    tmp = s / n
    new_sq = sq + (x - tmp) * (x - avg)
    new_v = new_sq / new_n
    return new_v, new_s, new_sq, new_n
\end{minted}
\caption{Welford's algorithm for online variance.}
\label{fig:python-var-welford}
\end{subfigure}
\vspace{-0.1in}
\caption{Offline and online algorithms for computing variance, implemented in Python.  In the online version, {\tt x} corresponds to the new stream element, and the first four are auxiliary parameters. We assume that all division operators are \emph{safe}, meaning that they produce a default value of $0$ when the denominator is $0$.}
\label{fig:python-var}
\end{figure}

In this section, we motivate our proposed technique with the aid of a motivating example for calculating statistical \emph{variance} for a list $xs$ of $n$ data points, defined as:
\begin{equation}\label{eq:variance}
\nu = \frac{\sum_{i=0}^n (xs[i] - \mu)^2}{n}
\end{equation}
where $\mu$ denotes the arithmetic mean of values in $xs$. 
Figure~\ref{fig:python-var-two-pass} shows the standard two-pass algorithm, written in Python, for implementing Eq.~\ref{eq:variance}. This algorithm first computes the mean in one pass over the input list $xs$ (lines 2-5) and then does a second pass to compute the squared differences from the mean (lines 7-10). To simplify this example and the presentation in the rest of the paper, we assume that all division operations are \emph{safe}, meaning that they produce a default value of zero if the denominator is zero.


In contrast, Figure~\ref{fig:python-var-welford} shows the Python implementation of its \emph{online version}, known as \emph{Welford's method}, based on a 1962 paper called \emph{Note on a Method for Calculating Corrected Sums of Squares and Products} \cite{welford}. The key insight behind Welford's method is the following pair of {recurrence relations} relating the new mean $\mu'$  and variance $\nu'$ to their previous values $\mu, \nu$:

\begin{equation}\label{eq:recurrence}
\begin{split}
\mu' & = \mu + \frac{(x - \mu)}{n} \\ 
\nu' & = \frac{\nu \times (n-1) + (x-\mu) \times (x-\mu')}{n}  \\
\end{split}
\end{equation}
Here, $x$ denotes the new element and $n$ denotes the total number of elements processed so far. 

Figure~\ref{fig:python-var-welford} uses these recurrence relations to compute the variance in an incremental way. Specifically, the online algorithm takes as input the new element {\tt x} to be processed and four previous computation results {\tt v}, {\tt s}, {\tt sq}, and {\tt n} where {\tt v} denotes the previously computed variance, {\tt s} denotes the sum of all previous elements, {\tt sq} denotes the previous sum of squared differences from the mean, and {\tt n} denotes the number of elements processed thus far.
It is easy to confirm that the implementation in Figure~\ref{fig:python-var-welford} computes the new variance as:
\[ 
\nu' = \frac{sq + (x - \frac{s}{n}) \times ({x - \frac{s+x}{n+1}})}{n+1}
\]
which is mathematically equivalent to Equation~\ref{eq:recurrence}. Thus, assuming that the initial auxiliary arguments of {\tt welford} are provided correctly (all zeros in this example), the online program yields the correct variance for a given stream of data points.
In the remainder of this section, we explain how our approach can automatically synthesize an implementation of Welford's algorithm given the offline version shown in Figure~\ref{fig:python-var-two-pass}.

\begin{figure}[!t]
\centering
\begin{subfigure}[b]{0.45\linewidth}
\begin{minted}
[
fontsize=\footnotesize,
escapeinside=||,
numbersep=5pt,
linenos
]
{haskell}
variance xs =
  let
          s = foldl (+) 0 xs
        avg = s / (length xs) 
    f acc x = acc + (x - avg)^2
  in (foldl f 0 xs) / (length xs)
\end{minted}
\caption{Two-pass algorithm for offline variance.}
\label{fig:ir-var-two-pass}
\end{subfigure}
~\qquad\qquad~
\begin{subfigure}[b]{0.45\linewidth}
\begin{minted}
[
fontsize=\footnotesize,
escapeinside=||,
numbersep=5pt,
linenos
]
{haskell}
welford (v, s, sq, n) x =
  let
     new_s = s + x
     new_n = n + 1
       avg = new_s / new_n
    new_sq = sq + (x - s / n) * (x - avg)
  in (new_sq / new_n, new_s, new_sq, new_n)
\end{minted}
\caption{Welford's algorithm for online variance.}
\label{fig:ir-var-welford}
\end{subfigure}
\vspace{-0.1in}
\caption{Intermediate presentation of variance computation.}
\vspace{-0.15in}
\label{fig:ir-var}
\end{figure}

\bfpara{Functional IR.}
While our tool, \tool, can take as input Python programs, it first converts the input to an intermediate functional representation that  facilitates synthesis.  Figure~\ref{fig:ir-var-two-pass} shows the corresponding intermediate representation of the two-pass variance computation, implemented using {\tt fold} operations. Specifically, line 3 computes the sum of all elements using {\tt fold} in the expected way, and line 4 computes the average as sum divided by length. Finally, line 6 computes variance using a {\tt fold} operation and the accumulator function {\tt f} declared on the previous line. As we will see later, this functional IR both simplifies presentation and also facilitates deductive reasoning. 

\begin{figure}[t]
\begin{center}
\begin{tabular}{ |c|l| }
\hline 
Parameter & Specification \\
\hline
\T{v} & \T{variance xs} \\
 \T{s} & \T{foldl (+) 0 xs} \\
 \T{sq} & \T{foldl (\acc x -> acc + (x - avg)^2) 0 xs} \\
 \T{n} & \T{length xs} \\
 \hline 
\end{tabular}
\end{center}
\vspace{-0.2in}
\caption{Relational Function Signature for Welford's algorithm} 
\vspace{-0.1in}\label{fig:rfs}
\end{figure}

\bfpara{RFS inference.}
Our approach does not take any inputs beyond the implementation of the offline algorithm, so it must first infer a suitable signature of the online program. Each online program takes as input the new element {\tt x} and the previous computation result {\tt v}; however, it may require additional inputs. To determine what auxiliary parameters may be required, \tool statically analyzes the input code to identify sub-expressions that are dependent on the input list and introduces a new parameter for each such sub-expression. In particular, Figure~\ref{fig:rfs} shows the four auxiliary parameters inferred for the variance algorithm, along with the sub-computations that they  represent. We refer to the mapping from Figure~\ref{fig:rfs} as a \emph{relational function signature} (RFS): the RFS maps each auxiliary argument of the online program to an expression $f(xs)$ in the offline program. 

\bfpara{Initializer.}
{ Recall that an online scheme consists of an initializer and an online algorithm. The initializer needs to handle the base case (i.e., empty list/stream), and it is easy to construct using the RFS. In particular, we can obtain suitable initial values for the {\tt v, s, sq, n} by evaluating the right-hand side expressions in Figure~\ref{fig:rfs} on an empty list. In this case, this yields $(0, 0, 0, 0)$ for the initializer for the online variance scheme.} 

\bfpara{RFS-guided synthesis methodology.}
Given a relational function signature like the one from Figure~\ref{fig:rfs}, \tool tries to synthesize an online program that is \emph{inductive relative to} this RFS. That is, assuming that the arguments of the online program are related to the original version as stipulated by the RFS, then the  values \emph{returned} by the online program should also satisfy the RFS. For example, consider the {\tt welford} implementation from Figure~\ref{fig:ir-var-welford}. The inductiveness of the RFS means that {\tt new\_v, new\_s, new\_sq}, and {\tt new\_n} should all satisfy the specification associated with {\tt v, s, sq, n} in Figure~\ref{fig:rfs} respectively. 

\bfpara{Syntax-directed sketch generation.}
To synthesize an online program that is inductive relative to the RFS, our approach generates a sketch from the offline program in a syntax-directed way.  For instance, for our running example, \tool generates the sketch shown in Figure~\ref{fig:sketch-var-sketch}, with the specification for each hole shown in Figure~\ref{fig:sketch-var-spec}. The key idea behind sketch generation is to  retain the reusable parts of the offline program, while replacing expressions that depend on the input list with holes. Furthermore, for each hole in the generated sketch, we can use the original expression as its specification and thereby decompose the synthesis problem into multiple independent sub-problems. For example, because the sketch shown in Figure~\ref{fig:sketch-var} contains three holes, each with its own specification, \tool can reduce the overall synthesis task to solving three \emph{independent} sub-problems.

\begin{figure}
\centering
\begin{subfigure}[b]{0.4\linewidth}
\begin{minted}
[
fontsize=\footnotesize,
escapeinside=||
]
{haskell}
online_variance (v, s, sq, n) x =
    let new_s = |$\hole_1$|
        new_n = |$\hole_2$|
        avg = s / new_n
        new_sq = |$\hole_3$| 
    in (new_sq / new_n, new_s, new_sq, new_n)
\end{minted}
\caption{Sketch.}
\label{fig:sketch-var-sketch}
\end{subfigure}
~\qquad\qquad~
\begin{subfigure}[b]{0.54\linewidth}
\small
\begin{tabular}{ |c|l| }
\hline 
{\small Unknown} & {\small Specification} \\
\hline
$\hole_1$ & \Tfn{foldl (+) 0 xs} \\
$\hole_2$ & \Tfn{length xs} \\
$\hole_3$ & \Tfn{foldl (\acc x -> acc + (x-avg)^2) 0 xs} \\
 \hline 
\end{tabular}
\caption{Specification.}
\label{fig:sketch-var-spec}
\end{subfigure}
\vspace{-0.3in}
\caption{Sketch generated by \tool for \T{variance}.}
\vspace{-0.1in}
\label{fig:sketch-var}
\end{figure}

\bfpara{Expression synthesis.}
Given the sketch shown in Figure~\ref{fig:sketch-var}, \tool aims to find an expression $e_i$ for each hole $\hole_i$ such that $e_i$ satisfies the corresponding specification for that hole.  As an example, let us consider the problem of synthesizing $\hole_1$ --- that is, we wish to find an expression over variables {\tt v, s, sq, n}, and {\tt x} such that the specification $f(xs)$ for $\hole_1$ is satisfied. But what does it even mean for an expression to satisfy its specification?  

To answer this question, consider the specification for $\hole_1$, and suppose the online algorithm has already processed elements {\tt xs} and that the new element to be processed is {\tt x}. Now,  we want to synthesize an expression $e_1$ for $\hole_1$ such that $e_1 = f(\texttt{xs ++ [x]})$ assuming that the function arguments satisfy the RFS. In this case, the RFS tells us that \texttt{s = foldl (+) 0 xs} which means we want to find an expression $e_1$ for $\hole_1$ such that $e_1$ is equivalent to {\tt foldl (+) 0 (xs ++ [x])}. Using the RFS and the semantics of {\tt fold}, we can show that the expression {\tt s + x} satisfies this specification because we have:
\[
(\texttt{s = foldl (+) 0 xs}) \  \Rightarrow \ (\texttt{s + x = foldl (+) 0 (xs ++ [x]))}
\]
Thus, \tool can infer that {\tt s + x} is a valid completion for $\hole_1$ in isolation  without having to reason about the rest of the sketch. 

But how does \tool find the expression {\tt s + x}? The simplest solution is  to perform enumerative search. While this works in simple cases,  the expressions that we need to synthesize can be quite large. For example, while the completion for $\hole_1$ is the simple expression {\tt s + x}, the completion for $\hole_3$ is a much more complex expression, namely {\tt sq + (x - s / n) * (x - avg)}. To deal with this challenge, \tool combines search with symbolic reasoning to derive expressions that are likely to be used in the target solution, as discussed later in Section~\ref{sec:expr-synthesis}.  
{
Intuitively, our key idea is to construct a logical formula in such a way that \emph{implicates} of this formula either directly correspond to the solution for the hole \emph{or} they can be used as useful building blocks when performing enumerative synthesis. As we show experimentally in Section~\ref{sec:eval}, this combination of symbolic reasoning and search is very beneficial in practice.
}

\bfpara{Summary.}
To summarize, \tool can {automatically} derive Welford's online algorithm for computing variance given {only} its standard two-pass (offline) implementation.  To do so, it first statically analyzes the offline implementation to learn a relational function signature, which drives the entire synthesis process. \tool also utilizes the offline program to generate a program sketch, which, along with the RFS, facilitates \emph{compositional} synthesis of each hole using a combination of search-based and symbolic methods. 

\section{Problem Statement}

In this section, we introduce the syntax and semantics of online and offline programs and formalize our problem statement.

\bfpara{Offline Programs.}
Figure~\ref{fig:ir-offline} shows the syntax of a simple functional language in which we express offline programs for batch processing. A program in this language takes as input a list $xs$ and evaluates an expression $E$.  Expressions include constants $c$, variables $x$, list expressions $L$, function applications  $g(E, \ldots, E)$ (where $g$ is either a built-in function or a user-defined lambda abstraction), and conditionals $\ite{E}{E}{E}$.  List expressions are formed using the standard list combinators \textsf{map}, \textsf{filter}, and \textsf{fold}, which may be arbitrarily nested. 
{Despite looking simple, this language is nevertheless Turing complete, and many batch processing programs are written in frameworks that support this style of functional programming~\cite{mapreduce, spark, flink}.}

\begin{example}\label{ex:offline}

Consider the following offline program $\lambda xs.~ \foldl(+, 0, xs) ~/~ \length(xs)$.
This program takes as input a list $xs$ of numbers and outputs their arithmetic mean. 

\end{example}
In the remainder of this paper, we assume a standard set of built-in functions such as the $+$ and \textsf{length} operators used in the previous example.  Given program $\prog$, we use the notation $\denot{\prog}_{l} = c$ to indicate that executing $\prog$ on list $l$ yields value~$c$.

\begin{figure}[!t]
\[
\begin{array}{rcl}
\text{Program } \prog &::=& \lambda xs. \ E \\
\text{ Expression } E &::=& c ~|~ x ~|~  L ~| ~ g(E, \ldots, E) ~|~ \ite{E}{E}{E} \\
\text{List Expr } L &::=& xs ~|~ \map(g, L) ~|~ \filter(g, L) ~|~ \foldl(g, E, L) \\
\text{Function } g &::=& \lambda \bar{x}. E ~|~ f \\
\end{array}
\]
\[
\begin{array}{c}
c \in \textbf{Constants} \quad x \in \textbf{Variables}  \quad xs \in \textbf{List Variables} \quad
 f \in \textbf{Built-in Functions} \\
\end{array}
\]
\vspace{-0.25in}
\caption{Syntax of the intermediate representation for offline programs.}
\vspace{-0.15in}
\label{fig:ir-offline}
\end{figure}

\begin{figure}[!t]
\[
\begin{array}{rcl}
\text{Scheme } \scheme &::=& (\init, \prog) \\
\text{Initializer } \init &::=& (c_1, \ldots,  c_n) \\
\text{Online program } \prog &::=& \lambda (y_1, \ldots, y_n). \lambda x. \ (E_1, \ldots, E_n) \\
\text{Expression } E &::=& c ~|~ x ~|  ~ y_i ~ | ~ g(E, \ldots, E) ~|~ \ite{E}{E}{E} \\
\text{Function } g &::=& \lambda \bar{z}. E ~|~ f \\
\end{array}
\]
\[
c \in \textbf{Constants} \quad x, y, z \in \textbf{Variables} \quad f \in \textbf{Built-in Functions}
\]
\vspace{-0.3in}
\caption{Syntax of the intermediate representation for online scheme.}
\vspace{-0.15in}
\label{fig:scheme-syntax}
\end{figure}

\bfpara{Online scheme.} As shown in Figure~\ref{fig:scheme-syntax}, an \emph{online implementation scheme} (or \emph{online scheme} for short) is a pair $\scheme = (\init, \prog')$  where $\init$, the  \emph{initializer}, is a tuple of constants $(c_1, \ldots, c_n)$, and $\prog'$ is a so-called \emph{online program}. 
Since the online program is expected to perform the same computation as the offline program but in an incremental fashion, it takes two arguments: (1) a tuple $(y_1, \ldots, y_n)$ which corresponds to the computational results over the previously processed stream elements, and (2) $x$, which corresponds to the new stream element to be processed. The return value of the online program is another tuple $(y_1', \ldots, y'_n)$, which corresponds to the new results after processing additional element $x$. Note that expressions in the online program are the same as their counterparts in offline programs except that  list combinators are disallowed to force incremental computation.


\begin{figure}[!t]
\[
\begin{array}{c}
\irulelabel
{}
{\denot{(\init, \prog)}_{\nil} = [\fst(\init)]}
{\textrm{(Lift-Nil)}}

\irulelabel
{\begin{array}{c}

    s, \init \vdash \prog \Downarrow s'
\end{array}}
{\denot{(\init, \prog)}_s = s'}
{\textrm{(Lift-Cons)}}

\\ \ \\

\irulelabel
{}
{\nil, \_ \vdash \prog \Downarrow \nil}
{\textrm{(S-Nil)}}

\irulelabel
{\begin{array}{c}
    \denot{\prog}_{h, c} = c' \quad
    t, c' \vdash  \prog \Downarrow s'
\end{array}}
{h:t, c \vdash \prog \Downarrow \fst(c') : s'}
{\textrm{(S-Cons)}}

\end{array}
\]
\vspace{-0.15in}
\caption{Semantics of the online scheme. $\denot{(I, P)}_s = s'$ means evaluating online scheme $(I, P)$ on stream $s$ yields stream $s'$, and $s, c \vdash (I, P) \Downarrow s'$ is an auxiliary relation, meaning online scheme $(I, P)$ evaluates to $s'$ given stream $s$ and current accumulators $c$.}
\vspace{-0.1in}
\label{fig:scheme-semantics}
\end{figure}

\bfpara{Semantics.}
 Figure~\ref{fig:scheme-semantics} presents the semantics for  online scheme $\scheme = (\init, \prog')$ using the notation $\denot{(\init, \prog')}_s = s'$, indicating that executing $\scheme$ on input stream $s$ yields another stream $s'$. To define the semantics, Figure~\ref{fig:scheme-semantics} uses an auxiliary relation of the form:
$  s, c \vdash \prog' \Downarrow s' 
$
that keeps track of the running accumulator $c$ (i.e., first argument of the online program). Given a non-empty stream $s$, the Lift-Cons rule in Figure~\ref{fig:scheme-semantics} initializes the accumulator to $\init$ and evaluates online program $\prog'$ on $s$ using the auxiliary rule S-Cons.  In particular, given a stream $s$ with head $h$ and tail $t$, S-Cons first evaluates online program $\prog'$ on $h$ and current accumulator $c$ and then recurses on tail $t$ with new accumulator values $c'$. Note that our convention in this paper is to designate the first element of the tuple returned by the online program to correspond to the result of the offline program. As such, S-Cons appends $s'$ to the first element of the tuple in $c'$.

\begin{example} \label{ex:online-prog}
The online scheme $\scheme$ for the arithmetic mean program Example~\ref{ex:offline}  consists of the initalizer $\init = (0, 0)$ and the following online program $\prog'$:
\[
\prog'((y, z), x) = \left( \left( y \cdot z + x \right) / (z + 1),~ z + 1 \right)
\]
Here, $y$ corresponds to the running mean and $z$ is the number of stream elements processed so far. Then, given the stream $s = [0, 1, 2, 3, \ldots ]$, we have:
\[
\denot{\scheme}_s = [0, 0.5, 1, 1.5, \ldots ]
\]
\end{example}

Next, we define what it means for an offline program to be equivalent to an online scheme:

\begin{definition}[\textbf{Online-Offline Equivalence}] \label{def:equiv}
Let $\prog$ be an offline program and $(\init, \prog')$ be an online scheme. We say that $(\init, \prog')$ is \emph{equivalent} to $\prog$, denoted $\prog \simeq (\init, \prog')$, if for any list $xs$ and its corresponding stream representation $xs'$, we have:
\[
 \denot{\prog}_{xs} = \last(\denot{(\init, \prog')}_{xs'})
\]
where $\last$ denotes the last element in a finite stream.
\end{definition}




\bfpara{Problem statement.}
Given an offline program $\prog$, our goal is to synthesize an online scheme $(\init, \prog')$ such that $\prog \simeq (\init, \prog')$.

\section{Methodology}\label{sec:method}

Before presenting our concrete synthesis algorithm, we first introduce the  general methodology and justify its correctness.  As stated earlier, our methodology hinges on the following notion of \emph{relational function signature (RFS)}:

\begin{definition}[\textbf{Relational Function Signature}] \label{def:rfs}
Let $\prog$ be an offline program with argument $xs$ and let  $\prog'$ be an online program with arguments $y, x$ where $y = (y_1, \ldots, y_n)$. A Relational Function Signature (RFS)  $\rfs$  maps each  $y_i$ to an offline expression $f_i(xs)$. We also write $\rfs(xs, y)$ to denote the formula $\bigwedge_{i=1}^n y_i = \rfs[y_i]$. 
\end{definition}

Intuitively, a relational function signature specifies the semantics of the additional arguments $y_1, \ldots, y_n$ of the online program in terms of expressions in the offline program. 

\begin{example}\label{ex:rfs}
Consider the offline program $\prog$ from Example~\ref{ex:offline} and the online program $\prog'$ from Example~\ref{ex:online-prog}. The relationship between $\prog$ and $\prog'$ is captured through the following RFS:
\[
\rfs[y] = \foldl(+, 0, xs) ~/~ \length(xs) \quad \quad \rfs[z] = \length(xs)
\]
Intuitively, this RFS states that the additional argument $y$ of the online program corresponds to the previous computation result and that $z$ keeps tracks of the number of list elements processed so far. 
\end{example}

Next, we introduce the notion of  \emph{inductiveness relative to a relational function signature}:



\begin{definition}[\textbf{Inductiveness relative to RFS}] \label{def:inductive-rfs}
Let $\rfs$ be an RFS between offline program $\prog$ (with argument $xs$) and  online program $\prog'$ (with arguments $y, x$).  We say that $\prog'$ is \emph{inductive relative to} $\rfs$ if and only if the following Hoare triple is valid:
\[
\hoare{\rfs(xs, {y})}{\quad {y}' := \prog'(y, x); \ xs' = xs \listappend [x] \quad \quad  } {\rfs(xs', {y}')}
\]
\end{definition}

Intuitively, an RFS is inductive if it is preserved after processing the next element in the input stream. That is, given an input stream $xs' = xs \listappend [x]$, if the RFS holds between $xs$ and $y$, then it should continue to hold between $xs'$ and $y'$ where $y'$ is the result of executing the online program $\prog'$ on the new stream element $x$ and previous computation results $y$. 

\begin{example}\label{ex:inductive-rfs}
Consider the offline and online programs from Examples~\ref{ex:offline} and~\ref{ex:online-prog}, and the RFS from Example~\ref{ex:rfs}. This RFS is inductive because:
\begin{enumerate}[leftmargin=*]
\item $\rfs[z]$ is preserved: assuming that $z$ is the length of $xs$, then $z'$ is $z+1$, which is the length of $xs \listappend [x]$. 
\item $\rfs[y]$ is preserved: assuming that $y$ is the arithmetic mean of $xs$, then $y'$ is computed as $(y \times z + x)/(z+1)$, which is indeed the arithmetic mean of $xs \listappend [x]$.
\end{enumerate}
 \end{example}

\begin{definition}[\textbf{Model of RFS}]\label{def:rfs-model} We say that an online scheme $\scheme = (\init, \prog')$ is a \emph{model} of RFS $\rfs$, denoted $\scheme \models \rfs$ if the following conditions are satisfied:
\begin{enumerate}[leftmargin=*]
\item  $\rfs(\nil, \init)$ evaluates to true, denoted $\init \models \rfs$
\item $\prog'$ is inductive relative to $\rfs$ 
\end{enumerate}
\end{definition}

\begin{example}\label{rfs-model}
Consider again the RFS $\rfs$ from Example~\ref{ex:inductive-rfs} and the online scheme from Example~\ref{ex:online-prog}. This online scheme is a model of $\rfs$ since it is inductive with respect to $\prog'$ (as shown in example~\ref{ex:inductive-rfs}), and  $\rfs[y] = \rfs[z] = 0$ on the empty list. Thus, the initializer $\init = (0, 0)$ also satisfies  $\init \models \rfs$.
\end{example}



We now state the following theorem that forms the basis of our synthesis methodology:

\begin{theorem}\label{thm:method-sound}
Let $\prog = \lambda xs. E$ be an offline program and $\scheme = (\init, \prog')$  an online scheme. Let $\rfs(xs, y)$ be an RFS between $\prog$ and $\prog'$ such that $\rfs[y_1] = E$. Then, if $\scheme \models \rfs$, we have $\prog \simeq \scheme$.
\end{theorem}

Proofs of all theorems can be found in the extended version of the paper.

\bfpara{Synthesis methodology.}
The previous theorem forms the basis of our synthesis algorithm. In particular, our synthesis methodology consists of three key steps:
\begin{enumerate}[leftmargin=*]
\item Given the offline program $\prog = \lambda xs. E$, find a  relational function signature $\rfs$ such that $\rfs[y_1] = E$.
\item Construct an initializer $\init$ such that $\init \models \rfs$.
\item Synthesize an online program $\prog'$ that is inductive relative to $\rfs$.
\end{enumerate}

If we can synthesize such a triple $(\rfs, \init, \prog')$ satisfying properties (1)-(3) from above, Theorem~\ref{thm:method-sound} guarantees that the resulting online scheme $\scheme = (\init, \prog')$ is equivalent to input  $\prog$.  Furthermore, we can also show that this methodology is complete under certain realistic assumptions:

\begin{theorem}\label{thm:rfs_completeness}
Let $\prog = \lambda xs. E$ be an offline program and let $\scheme = (\init, \prog')$ be an online scheme such that $\prog \simeq \scheme$. If the expression $\foldl(\prog', \init, xs)$ has an inductive invariant $\lambda xs. \lambda y. \phi$ where $\phi \equiv \bigwedge_i y_i = E_i$ with $E_1 = E$, then  there exists an RFS satisfying conditions (1) -- (3) of our methodology. 
\end{theorem}

{To gain some intuition about this theorem, we first observe that for $\prog$ and  $(\init, \prog')$ to be equivalent, we must have $
\prog(xs) = \fst(\foldl(\prog', \init, xs))
$
for any input list $xs$. Hence, at the very least, we must have $y_1 = E$ as an invariant of $\prog'$, where $E$ is the body of the offline program. However, since it may not be an \emph{inductive} invariant, we may need to logically strengthen it to make it inductive.  The theorem states that, as long as the required strengthening is of the form $\bigwedge_{i=2}^n y_i = E_i$ (where $y_i$'s are the auxilary arguments of the online program and $E_i$ is an offline expression), then the synthesis methodology is also complete. This is a very mild assumption that also underlies other prior work on incremental computation \cite{annieliu, ingress}. } Intuitively, we can find an inductive invariant because online programs maintain an auxiliary state that is always equivalent to some computation result, so the invariant can be expressed as a conjunction of equalities.

\section{Synthesis Algorithm}\label{sec:synthesis}

In this section, we describe our synthesis algorithm based on the methodology introduced in the previous section. 

\subsection{Top-Level Algorithm}

\begin{figure}[!t]
\begin{algorithm}[H]
\caption{Online Scheme Synthesis}
\label{algo:toplevel}
\begin{algorithmic}[1]
\Procedure{\textsc{Synthesize}}{$\prog$}
\vspace{2pt}
\Statex \textbf{Input:} An offline program $\prog$
\Statex \textbf{Output:} An equivalent online scheme $(\init, \prog')$
\State $\rfs \gets \textsc{ConstructRFS}(\prog)$
\State $\init \gets \mathsf{Model}(\rfs[xs \mapsto \nil])$
\State $\prog' \gets \textsc{SynthesizeOnlineProg}(\prog, \rfs)$

\State\Return $(\init, \prog')$
\vspace{2pt}
\EndProcedure
\end{algorithmic}
\end{algorithm}
\vspace{-0.3in}
\end{figure}

Our top level synthesis procedure is presented in Algorithm~\ref{algo:toplevel} and follows the methodology from Section~\ref{sec:method}. Specifically, it first constructs an RFS by analyzing the offline program  (line 2). It then synthesizes the initializer by replacing each occurrence of $xs$ in $\rfs$ with $\nil$ and obtaining a model of the resulting formula (line 3). Finally,  it invokes the {\sc SynthesizeOnlineProg} procedure to construct an online program $\prog'$ such that $\prog'$ is inductive relative to $\rfs$ (line 4). 

\begin{figure}[!t]
\vspace{-10pt}
\begin{algorithm}[H]
\caption{Learning RFS}
\label{algo:rfs}
\begin{algorithmic}[1]
\Procedure{\textsc{ConstructRFS}}{$\prog$}
\vspace{2pt}
\Statex \textbf{Input:} An offline program $\prog = \lambda xs. E$
\Statex \textbf{Output:} A relational function signature $\rfs$
\State $\rfs \gets \set{y_1 \mapsto E}$
\For {$e_2, \ldots, e_n \in \mathsf{ListExpr}(E)$}
    \State $\rfs \gets \rfs[y_i \mapsto e_i]$
\EndFor
\State \Return $\rfs$
\vspace{2pt}
\EndProcedure
\end{algorithmic}
\end{algorithm}
\vspace{-0.3in}
\end{figure}

Algorithm~\ref{algo:rfs} presents our technique for constructing a relational function signature. The key idea underlying {\sc ConstructRFS} is the following: For any expression $e$ of $\prog$ that performs some operation over the input list $xs$, the online program \emph{may} require an additional argument to store the previous computation result. Thus, {\sc ConstructRFS} iterates over list expressions $e_2, \ldots, e_n$ in the offline program and introduces a new argument $y_i$ for each $e_i$, 
with the corresponding mapping $\rfs[y_i] = e_i$. Here (and in the remainder of the paper), we use the term ``list expression" to mean any expression that has $xs$ as a child in the abstract syntax tree of $\prog$. Since our convention is to store the result of the offline program in $y_1$, note that line 2 of Algorithm~\ref{algo:rfs} maps $y_1$ to $E$, which is the body of the offline program.

\bfpara{Remark.}
The {\sc ConstructRFS} procedure may end up introducing more accumulators (i.e., auxiliary parameters) than necessary. If the synthesized online program does not end up using them, our implementation removes such unnecessary variables from the signature of the online program in a subsequent post-processing step.

\begin{example}\label{ex:rfs}
Consider the offline program from Example~\ref{ex:offline}. Our algorithm produces the following relational function signature:
\[
\{ y_1 \mapsto \foldl(+, 0, xs) ~/~ \length(xs), y_2 \mapsto   \length(xs), y_3 \mapsto \foldl(+, 0, xs) \}
\]
\end{example}

\subsection{Synthesis of Online Programs} \label{subsec:syn_online_prog}
\begin{figure}[!t]
\begin{algorithm}[H]
\caption{Online Program Synthesis}
\label{algo:online-prog}
\begin{algorithmic}[1]
\Procedure{\textsc{SynthesizeOnlineProg}}{$\prog, \rfs$}
\vspace{2pt}
\Statex \textbf{Input:} Offline program $\prog$, relational function signature $\rfs$
\Statex \textbf{Output:} An online program $\prog'$ such that $\rfs$ is inductive with respect to $\prog'$
\State $\prog', \context \gets \textsc{Decompose}(\rfs, \prog)$ 
\For{\textbf{each} $h \in \textsf{Holes}(\prog')$}
    \State $E \gets \textsc{SynthesizeExpr}(\rfs, \context[h])$
    \If{$E = \bot$}  \Return $\bot$
    \EndIf
    \State $\prog' \gets \prog'[E / h]$
\EndFor
\State\Return $\prog'$
\vspace{2pt}
\EndProcedure
\end{algorithmic}
\end{algorithm}
\vspace{-0.3in}
\end{figure}

Algorithm~\ref{algo:online-prog} presents our approach for synthesizing online programs. As mentioned in Section~\ref{sec:intro}, the main idea is to \emph{decompose}  the synthesis task into several subproblems that can be solved \emph{completely independently}. In particular, the algorithm performs this decomposition by first generating a program sketch, where each hole represents an independent synthesis task with its own specification (line 2). The loop in lines 4--6 then solves each  sub-problem by calling the {\sc SynthesizeExpr} procedure (line 4). In the remainder of this section, we describe our decomposition technique and expression synthesis algorithm in more detail.

\subsubsection{\bf Decomposition}

\begin{figure}[!t]
\small
\[
\begin{array}{c}

\irulelabel
{\begin{array}{c}
\dom(\rfs) = \set{y_1, \ldots, y_n} \\ 
\rfs \vdash \rfs[y_1] \rewrite \sketch_1, \context_1 \quad \ldots \quad
\rfs \vdash \rfs[y_n] \rewrite \sketch_n, \context_n \\
\end{array}
}
{\rfs \vdash \lambda xs. E \rewrite \lambda (y_1, \ldots, y_n). \lambda x. (\sketch_1, \ldots, \sketch_n),~ \context_1 \cup \ldots \cup \context_n }
{\textsc{(Prog)}}

\quad

\irulelabel
{\begin{array}{c}
\textsf{LeafExpr}(E) \\
\textsf{Type}(E) \neq \textsf{List}
\end{array}}
{\rfs \vdash E \rewrite E, \emptymap}
{\textsc{(Leaf)}}

\\ \ \\

\irulelabel
{\Box = \textsf{Hole}(L)}
{\rfs \vdash L \rewrite \Box, \set{\Box \mapsto L}}
{\textsc{(List)}}

\quad

\irulelabel
{\begin{array}{c}
\rfs \vdash E_1 \rewrite \sketch_1, \context_1 \quad
\ldots \quad
\rfs \vdash E_n \rewrite \sketch_n, \context_n \\
\end{array}}
{\rfs \vdash g(E_1, \ldots, E_n) \rewrite g(\sketch_1, \ldots, \sketch_n),~ \context_1 \cup \ldots \cup \context_n}
{\textsc{(Func)}}

\\ \ \\

\irulelabel
{\begin{array}{c}
\rfs \vdash E_1 \rewrite \sketch_1, \context_1 \quad
\rfs \vdash E_2 \rewrite \sketch_2, \context_2 \quad
\rfs \vdash E_3 \rewrite \sketch_3, \context_3 \\
\end{array}}
{\rfs \vdash \ite{E_1}{E_2}{E_3} \rewrite \ite{\sketch_1}{\sketch_2}{\sketch_3},~ \context_1 \cup \context_2 \cup \context_3}
{\textsc{(ITE)}}

\end{array}
\]
\vspace{-0.1in}
\caption{Rules for decomposition.}
\label{fig:decomp}
\vspace{-0.1in}
\end{figure}

We present our decomposition technique using inference rules of the form $\rfs \vdash E \rewrite \sketch, \context $ where $\sketch$ is an expression with holes (i.e., \emph{sketch}) and $\context$ is a mapping from each hole to its corresponding specification. Here, the specification is an expression in the offline program, and the goal of the subsequent synthesis task is to generate an \emph{online expression} $e$ for each hole $h$ such that $e$ is equivalent to $\context[h]$ modulo the RFS. 

Before we go into the details of our sketch generation procedure, we first provide some high-level intuition. The key idea is to replace expressions that directly operate over the input list  with holes but reuse the general high-level structure of the offline algorithm.  For example, consider the expression $e$ given by
\texttt{foldl(+, 0, xs) / length(xs)}. If we have a way of incrementally computing \texttt{foldl(+, 0, xs)} and \texttt{length(xs)} using expressions $e_1$ and $e_2$ respectively, we can also incrementally compute $e$ as $e_1/e_2$. Thus, our decomposition technique implicitly assumes that the online program can be obtained by \emph{composing} incremental computations over list expressions using operators that already appear in the offline program. While this assumption could \emph{in principle} be violated (thereby causing our synthesis algorithm to lose completeness), we have, \emph{in practice}, not encountered any cases violating this assumption.

Figure~\ref{fig:decomp} presents our decomposition algorithm as inference rules. The first rule labeled \textsc{Prog} utilizes the RFS to generate a sketch for the entire program. In particular, if there are $n$ variables in the domain of the RFS, the body of the online program consists of an $n$-ary tuple $(\sketch_1, \ldots, \sketch_n)$ where each sketch $\sketch_i$ corresponds to $\rfs[y_i]$.  The next rule, labeled {\sc Leaf}, is used to ``copy over" shared expressions that belong to the syntax of both online and offline programs. The rule labeled {\sc List} introduces holes: Since list expressions $L$ are disallowed in online programs, they must be synthesized from scratch, and the resulting expression must be semantically equivalent to expression $L$ in the offline program. Thus,  this rule states that the specification for the introduced hole is $L$. The remaining rules are used to recursively construct sketches for compound expressions.  For example, given an expression $g(E_1, \ldots, E_n)$ the {\sc Func} rule constructs a sketch $g(\sketch_1, \ldots, \sketch_n)$ by recursively constructing sketches for each $E_i$. 

\begin{example}
Consider the RFS from Example~\ref{ex:rfs} and the offline program from Example~\ref{ex:offline}. Our decomposition procedure generates the following program sketch for the online program:
\[
\lambda (y_1, y_2, y_3). \lambda x. \ (\hole_1 / \hole_2, \hole_2, \hole_1)
\]
and the specifications of each hole are as follows:
\[
\{ \hole_1 \mapsto \foldl(+, 0, xs), \ \hole_2 \mapsto \length(xs) \}
\]
Thus, the decomposition produces two independent synthesis tasks.
\end{example}

\subsubsection{\bf Expression Synthesis}\label{sec:expr-synthesis}

The goal of expression synthesis is to find an online expression $E'$ that is equivalent to an offline expression $E$ modulo the RFS. Thus, before we discuss our synthesis algorithm, we first introduce the concept of \emph{equivalence modulo RFS}:

\begin{definition}{\bf (Equivalence modulo RFS)}\label{def:eq-expr}
We say that an offline expression $E$ is equivalent to online expression $E'$ modulo the RFS iff: 
\[
\rfs(xs, y) \models E' = E[(xs\pp[x])/xs]
\]
\end{definition}
In other words, an online expression $E'$ is equivalent to $E$ if we can show that $E'=E[(xs\pp[x])/xs]$ under the assumption that the RFS holds. To gain further intuition about this definition, recall that $xs$ denotes the previously processed elements and $x$ is the new element, so the elements processed so far correspond to the list $xs \pp [x]$. This is why $E'$ should be equivalent to $E$ after substituting $xs$ (the argument of offline program) with $xs\pp [x]$. Furthermore, since the RFS gives the mapping between the auxiliary variables of the online program and sub-expressions in the offline program,  equality between $E$ and $E'$ only makes sense when we utilize the mapping given by the RFS.

\begin{example}\label{ex:eq-expr}
Consider the RFS $\rfs$ from Example~\ref{ex:inductive-rfs}.  Then, the online expression $(y_1 \times y_2)+x$ is equivalent to $\foldl(+, 0, xs)$ modulo $\rfs$ because:
\[
y_1 = \foldl(+, 0, xs)~/~\length(xs) \land y_2 = \length(xs) \models \foldl(+, 0, xs\pp[x]) = (y_1 \times y_2) + x
\]
\end{example}

\begin{figure}[!t]
\[
\begin{array}{lll}
\foldl(g, c, xs\pp [x]) & = &  g(\foldl(g, c, xs), x) \\
\map(g, xs\pp[x]) & = & \map(g, xs) \pp [g(x)] \\
\filter(g, xs \pp [x]) & = & g(x) \ ? \  \filter(g, xs) \pp [x] : \filter(g, xs)

\end{array}
\]
\vspace{-0.2in}
\caption{Axioms involving higher-order combinators.}
\label{fig:axioms}
\vspace{-0.2in}
\end{figure}

\begin{figure}[!t]
\begin{algorithm}[H]
\caption{Expression synthesis algorithm}
\label{algo:expr-synth}
\begin{algorithmic}[1]
\Procedure{\textsc{SynthesizeExpr}}{$\rfs, E$}
\vspace{2pt}
\Statex \textbf{Input:} Relational function signature $\rfs$, offline expression $E$
\Statex \textbf{Output:} Online expression $E'$
\State $\chi \gets \textsc{FindImplicate}(\rfs, E[(xs\pp[x])/xs] = \hole)$
\If{$\chi$ \textsf{matches} $\hole = E'$} \State \Return $E'$
\Else
\State $\theta \gets {\textsc{MineExpressions}}(\rfs, E)$
\State \Return \textsf{EnumSynthesize}($\rfs, E, \theta$)
\EndIf
\vspace{2pt}
\EndProcedure
\vspace{2pt}

\Procedure{\textsc{FindImplicate}}{$\rfs, T$}
\vspace{2pt}
\Statex \textbf{Input:} Relational function signature $\rfs$, implicate template $T$
\Statex \textbf{Output:} Implicate of $\rfs$
\State $\axioms \gets \textsc{AddAxioms}(\Phi)$
\State $\psi \gets \rfs \land T \land \bigwedge_i \axioms_i$
\State $(\psi', V) \gets \mathsf{ReplaceListExprs}(\psi)$
\State \Return $\mathsf{ElimQuantifier}(\exists V. \psi')$
\vspace{2pt}
\EndProcedure

\vspace{2pt}

\Procedure{\textsc{MineExpressions}}{$\rfs, E$}
\vspace{2pt}
\Statex \textbf{Input:} RFS $\rfs$, offline expression $E$, unrolling depth $k$ (hyperparameter)
\Statex \textbf{Output:} Set of terms that are likely to be useful in enumerative synthesis
\State $\varphi \gets \textsf{True};$  \quad $(E', V) \gets \mathsf{Unroll}(E, k+1)$
\For{$(y_i, E_i) \in \rfs$}
\State $(E_i', V_i) \gets \mathsf{Unroll}(E_i, k);$
$\quad \varphi \gets \varphi \land (y_i = E_i')$;  $\quad V \gets V \cup V_i$
\EndFor

\State $\psi \gets \mathsf{ElimQuantifier}(\exists V. \ (\varphi \land \hole = E'))$
\State \Return $\{ \mathsf{Templatize}(t) \ | \ (\hole = t) \in \mathsf{Literals}(\psi)  \} $
\vspace{2pt}
\EndProcedure
\end{algorithmic}
\end{algorithm}
\vspace{-0.2in}
\end{figure}

Algorithm~\ref{algo:expr-synth} presents our expression synthesis algorithm for finding an online expression $E'$ that is equivalent to offline expression $E$ modulo the RFS $\rfs$. The basic idea is to use symbolic reasoning to find an \emph{implicate} of $\rfs$ that is of the form $E' = E[(xs\pp [x])/xs]$  where $E'$ is a term over variables $x, y_1, \ldots, y_n$. By definition, an implicate of a formula is implied by it; thus, if we can find an implicate of $\rfs$ of this form, it satisfies Definition~\ref{def:eq-expr} by construction. However, the key challenge is that both the RFS $\rfs$ and offline expression $E$ contain higher-order combinators such as $\foldl$ and $\map$, so it is not immediately obvious how to use an SMT solver to find a suitable implicate.

Our core approach to solving this problem is summarized in the {\sc FindImplicate} procedure in Algorithm~\ref{algo:expr-synth}. This algorithm takes as input the RFS $\rfs$ and an implicate template $T$, and computes an instantiation of $T$ that is implied by $\rfs$ as follows:
\begin{enumerate}[leftmargin=*]
\item  First, it adds axioms that relate the result of applying a higher-order combinator to $xs\pp[x]$ to the result of applying the combinator to $xs$  (line 9). Figure~\ref{fig:axioms} shows a set of axiom schema that are instantiated based on the specific terms used in $\rfs$.
\item Next, it  constructs a formula that is the conjunction of $\rfs, T$, and all the axioms $\axioms$ generated in the previous line. 
\item Third, it replaces each list expression with a fresh variable by calling the \textsf{ReplaceListExprs} procedure at line 11. The idea is to eliminate higher-order combinators like map and fold after adding all relevant axioms about them. Here, \textsf{ReplaceListExprs} returns a new formula $\psi'$ and a set of variables $V$ introduced by this transformation. 
\item Finally, it uses  quantifier elimination  to obtain a formula over variables $x, y_1, \ldots, y_n$. 
\end{enumerate}
We illustrate this procedure through a simple example: 

\begin{example}
Consider the RFS $\rfs$ and offline expression  $E$ from Example~\ref{ex:eq-expr} where:
\[
\begin{array}{lll}
\rfs & \equiv &  y_1 = \foldl(+, 0, xs)~/~\length(xs) \land y_2 = \length(xs) \\
T & \equiv & \hole = \foldl(+, 0, xs\pp [x])
\end{array}
\]
For this example, there is only one relevant axiom, namely:
\[
\foldl(+, 0, xs\pp[x]) = \foldl(+, 0, xs) + x
\]
After replacing list expressions with fresh variables, we obtain the following formula:
\[
y_1 = v_1/v_2 \land y_2 = v_2 \land v_3 = v_1 + x\land \hole = v_3 
\]
where $v_1, v_2, v_3$ represent $\foldl(+, 0, xs)$, $\length(xs)$, and $\foldl(+, 0, xs\pp[x])$ respectively. Finally, after eliminating $v_1, v_2, v_3$ from this formula, we obtain:
\[
\hole = (y_1 \times y_2) + x
\]
Hence, given the expression $\foldl(+, 0, xs)$, {\sc SynthesizeExpr} returns $(y_1 \times y_2) + x$ as the equivalent online expression.
\end{example}

If  {\sc FindImplicate} returns an equality of the form $\hole = E'$ (line 3 in Algorithm~\ref{algo:expr-synth}), then $E'$ is the equivalent online expression for $E$, so the algorithm returns $E'$ at line 4. However, {\sc FindImplicate} may not always return such a formula because, for example, the added axioms may not be sufficient to adequately capture the semantics of all list expressions. In this case, the {\sc SynthesizeExpr} algorithm falls back on enumerative synthesis (line 7) but leverages the insights from {\sc FindImplicate} to \emph{mine} useful expressions that can be used as building blocks. In particular, given the RFS $\rfs$ and offline expression $E$, the {\sc MineExpressions} procedure returns a set of \emph{templatized} expressions that are likely to be useful for enumerative synthesis.

The basic idea behind {\sc MineExpressions} is the same as {\sc FindImplicates}; however, rather than adding axioms about the higher-order combinators, it simply \emph{unrolls} them: That is, given an offline expression $E$ over list $xs$, the procedure \textsf{Unroll} instantiates $xs$ with a symbolic list of size $k$ and symbolically executes $E$ on this list.  Thus, the formula $\varphi$ in the {\sc MineExpressions} algorithm corresponds to an unrolled version of $\rfs$ on lists of size $k$, and $E'$ corresponds to an unrolled version of $E$ on a list of size $k+1$. As in the {\sc FindImplicates} procedure, we use quantifier elimination to find an implicate of the formula $\varphi \land \hole = E'$ over variables $x, y_1, \ldots, y_n$. However, because $\varphi$ and $E'$ are essentially under-approximations of $\rfs$ and $E$ respectively, the resulting formula may  not be a valid implicate. Thus, our synthesis algorithm simply mines \emph{templatized expressions} from the resulting formula by replacing constants, which are typically the root cause for the formula not being a valid implicate,  with holes. These templatized expressions are then added to the grammar for online expressions to expedite enumerative synthesis at line 7. This \textsf{EnumSynthesize} procedure is based on basic top-down enumerative search and checks correctness using testing (see  Section~\ref{sec:impl}).

\begin{example}
{Consider the RFS $\Phi$ and offline expression $\rfs[\mathrm{sq}]$ from Figure~\ref{fig:rfs} in Section~\ref{sec:overview} where:}
\begin{align*}
    \rfs &\equiv \textrm{sq} = \foldl(\lambda c. \lambda x.~ c+(x-\textrm{avg})^2, 0, xs)  \ldots \\
    T &\equiv \hole = \foldl(\lambda c. \lambda x.~ c+(x-\textrm{avg}')^2, 0, xs\pp[x])
\end{align*}
For this example, there is only one relevant axiom, namely
\[
\foldl(\lambda c. \lambda x.~ c+(x-\textrm{avg})^2, 0, xs\pp[x]) = \foldl(\lambda c. \lambda x.~ c+(x-\textrm{avg})^2, 0, xs) + (x-\textrm{avg})^2
\]
After replacing list expressions with fresh variables, we obtain the following formula for $\rfs \wedge T$:
\[
 \textrm{sq} = v_2 \wedge \hole = v_3,
\]
After eliminating the fresh variables, we obtain \emph{true} as an implicate, which is not useful. Hence,  in line 6 of Algorithm \ref{algo:expr-synth}, we mine expressions by instantiating $xs$ with a symbolic list of size $k$ and symbolically execute $\phi[\textrm{sq}]$ on this list.
When $k = 3$ and $xs = [x_1, x_2, x_3]$, we have the following after executing line 14-16 of Algorithm \ref{algo:expr-synth}:
\begin{align*}
    \rfs &\equiv \textit{sq} = (x_1 - \textrm{avg})^2 + (x_2 - \textrm{avg})^2 + (x_3 - \textrm{avg})^2 \wedge n = 3 \wedge \ldots \\
    T &\equiv \hole = (x_1 - \textrm{avg}')^2 + (x_2 - \textrm{avg}')^2 + (x_3 - \textrm{avg}')^2 + (x - \textrm{avg}')^2 \\
    \varphi &\equiv \exists x_0, x_1, x_2. ~\rfs \wedge T
\end{align*}
where we introduced $\textrm{avg} = \frac{1}{3}(x_1 + x_2 + x_3)$ and $\textrm{avg}' = \frac{1}{4}(x_1 + x_2 + x_3 + x)$ to simplify presentation.
Finally, running quantifier elimination gives the following expression:
\[
\hole = \frac{1}{12}(s^2 - 6\cdot s\cdot x + 12\cdot\textrm{sq} + 9\cdot x^2),
\]

After replacing constants with holes, we obtain the following template:
\[
\frac{s^2 - \ ??_1*s*x + \ ??_2*\textrm{sq} \  + \ ??_3*x^2}{??_4}
\]

Note that the desired expression, which is:

\[
\frac{s^2 - (2n)*s*x + (n(n+1))*\textrm{sq} + (n^2)*x^2}{n(n+1)}
\]
can be obtained from this template by replacing the unknowns with  expressions $2n$, $n(n+1)$,  $n^2$, and $n(n+1)$ respectively.
Hence, obtaining such templates via {\sc MineExpressions} ends up significantly speeding up enumerative synthesis.  

\end{example}

\begin{theorem} \label{thm:synthesizeexpr_soundness}
If {\sc SynthesizeExpr}($\rfs, E$) returns $E'$, then $E'$ is indeed equivalent to $E$ modulo $\rfs$.
\end{theorem}
Finally, we conclude this section by stating the soundness of the end-to-end synthesis procedure:
\begin{theorem} \label{thm:synthesize_soundness}
If {\sc Synthesize}($\prog$) returns $(\init, \prog')$, then we have $\prog \simeq (\init, \prog')$.
\end{theorem}
\section{Implementation}\label{sec:impl}

We have implemented our proposed technique in a tool called \tool written in Python. \tool uses the Reduce computer algebra system \cite{reduce} to perform quantifier elimination for both linear and nonlinear integer and rational arithmetic. When invoking Reduce, \tool ensures that formulas belong to a theory that admits quantifier elimination by replacing foreign terms with fresh variables. 

\bfpara{Conversion to functional IR.}
As mentioned earlier, \tool operates over offline programs written in a functional IR with higher-order combinators. However, \tool can also take as input Python programs and automatically converts them to our intermediate representation using a set of syntax-directed translation rules. Since transpilation from imperative to functional languages is an orthogonal problem, we refer the interested reader to prior papers on this topic~\cite{ben-oopsla22}.

\bfpara{Handling additional arguments.}
While our technical section assumes that the offline program takes a single list $xs$ as an argument, real-world programs can take additional arguments.  In this case, the RFS constructed by \tool includes those additional arguments and assumes a one-to-one correspondence between the additional arguments of the offline and online programs.

\bfpara{Solving templates via polynomial interpolation.}
Recall from Section \ref{sec:expr-synthesis} that {\sc MineExpressions} returns a set of templates (expressions with unknowns), which are utilized when performing enumerative search. However, there are several cases where the unknowns in these templates can be directly solved for using polynomial interpolation~\cite{polyinterp}. In particular, if the online procedure takes an auxiliary parameter $n$ that represents the number of processed stream elements i.e, the length of the list, then the desired expression can oftentimes be obtained by instantiating the unknowns in the templates with univariate polynomials over $n$. 
\tool utilizes SciPy's interpolation library~\cite{scipy} to infer candidate univariate polynomials and checks whether the synthesized expression is equivalent to its offline version modulo the RFS. If it is not, \tool falls back upon enumerative search using the generated template. We refer the interested reader to the extended version of the paper for more details.

\bfpara{Checking Equivalence modulo RFS.}
Ideally \tool would check equivalence between the online and offline expressions over all possible input streams. However, since automatically checking equivalence is out of scope for existing techniques, \tool resorts to unsound equivalence checking methods based on testing and bounded verification.  However, in practice, we have not come across any cases where the equivalence checker yielded an incorrect result.
\section{Evaluation}\label{sec:eval}

In this section, we evaluate \tool through experiments that aim to answer the following research questions:

\begin{enumerate}[label = \textbf{RQ\arabic*.}]
    \item (Usefulness) Can \tool convert non-trivial offline programs into equivalent online schemes?
    \item (Comparison against existing tools) How does \tool compare against state-of-the-art general purpose synthesizers like CVC5 \cite{cvc5} and Sketch \cite{sketch,sketch2}?
    \item (Ablation) How important are the key ideas underlying our approach?
\end{enumerate}

\begin{table}[] 
\caption{Statistics about the benchmark set} \label{tab:benchmark_stats}
\vspace{-10pt}
\begin{tabular}{|c|cc|cc|}
\hline
\multirow{2}{*}{} & \multicolumn{2}{c|}{Avg. AST Size}    & \multicolumn{2}{c|}{Median AST Size}  \\ \cline{2-5} 
                  & \multicolumn{1}{c|}{Offline} & Online & \multicolumn{1}{c|}{Offline} & Online \\ \hline
Stats             & \multicolumn{1}{c|}{25}      & 45     & \multicolumn{1}{c|}{24}      & 39     \\ \hline
Auction           & \multicolumn{1}{c|}{79}      & 76     & \multicolumn{1}{c|}{42}      & 44     \\ \hline
\end{tabular}
\end{table}

\bfpara{Sources of benchmarks.}
To answer these questions, we collected benchmarks from two domains where online algorithms  play a key role: %
\begin{itemize}[leftmargin=*]
    \item {\bf Statistics.} Online algorithms are particularly important in the context of \emph{statistical computations} over streaming data.  To evaluate \tool in this context, we collected 34 batch-processing programs that perform statistics over a list of elements. These functions are taken from two sources: The first is SciPy \cite{scipy},\footnote{Many SciPy functions use external libraries, such as numpy, for numerical computations. Since our prototype \tool does not support such external libraries, we manually pre-processed those  benchmarks.} an open-source Python library used for scientific computing, and the second one is OnlineStats.jl \cite{onlinestats}, a popular open-source Julia library that implements useful single-pass algorithms. Since the Julia benchmarks are online programs, we manually wrote their offline version in Python.  These statistics benchmarks include functions for computing skewness~\cite{sandia-report}, standard error of the mean (SEM)~\cite{sem}, geometric mean, LogSumExp, etc. 
    %
    \item {\bf Auctions.} Another domain where online algorithms play an important role is \emph{online auctions} that involve queries over continuous data streams. To evaluate \tool in this context, we consider 18 queries from the Nexmark benchmark suite, which includes queries that commonly arise in online auctions~\cite{nexmark}.\footnote{While there 23 queries in the Nexmark benchmark suite, 5 of them require mini-batching, which we currently do not support, so we consider 18 out of these 23 benchmarks. Furthermore, since all of these queries are written for streaming data, we manually wrote their batch processing versions.} Example tasks from this benchmark suite include generating bidding statistics reports, monitoring new bidders, determining top-$k$ bids, etc. 
\end{itemize}

\bfpara{Obtaining ground truth schemes.}
Some of the benchmarks in our suite contain both the offline program and its corresponding online implementation. For offline programs whose corresponding online version was not available, we either found its (established) online version from a different source or wrote it  ourselves. 


\bfpara{Statistics about benchmarks.}
Table~\ref{tab:benchmark_stats} provides statistics about these benchmarks in terms of the average and median program size, where size is measured in terms of the number of nodes in the abstract syntax tree (AST). While the size of the offline and online programs are similar for the auction benchmarks, we note that the size of the online programs are significantly larger (1.7$\times$) on average for the statistics benchmarks. We also note that some of these benchmarks require synthesizing very complex expressions (up to size $96$).

\bfpara{Experimental setup.}
All of our experiments are conducted on a machine with an Apple M1 Pro CPU and 32 GB of physical memory, running the macOS 14.1 operating system. For each task, we set the timeout to 10 minutes.

\subsection{Main Results}
%
%
%

\begin{table}[]
\caption{Main synthesis result.} \label{table:eval_main_res}
\vspace{-10pt}
\begin{tabular}{|c|cc|cc|}
\hline
\multirow{2}{*}{}    & \multicolumn{2}{c|}{Stats}                     & \multicolumn{2}{c|}{Auction}                   \\ \cline{2-5} 
                     & \multicolumn{1}{c|}{\% solved} & Avg. Time (s) & \multicolumn{1}{c|}{\% solved} & Avg. Time (s) \\ \hline
\tool & \multicolumn{1}{c|}{97\%}      & 33.4          & \multicolumn{1}{c|}{100\%}     & 10.0          \\ \hline
Sketch               & \multicolumn{1}{c|}{12\%}      & N/A           & \multicolumn{1}{c|}{17\%}      & N/A           \\ \hline
CVC5                 & \multicolumn{1}{c|}{36\%}      & N/A           & \multicolumn{1}{c|}{39\%}      & N/A           \\ \hline
\end{tabular}
\end{table}

\begin{figure}[!t]
\centering
\begin{subfigure}[b]{0.45\linewidth}
\includegraphics[width=\textwidth]{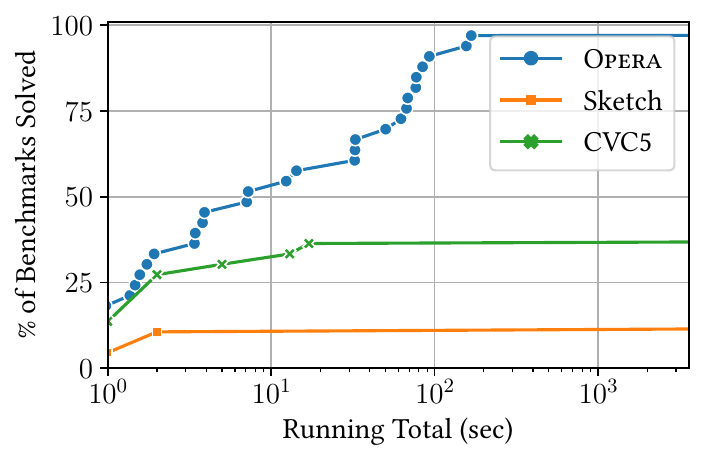}
\vspace{-0.2in}
\caption{\% of Stats Benchmarks Solved by Time}
\end{subfigure}
~\qquad\qquad~
\begin{subfigure}[b]{0.45\linewidth}
\includegraphics[width=\textwidth]{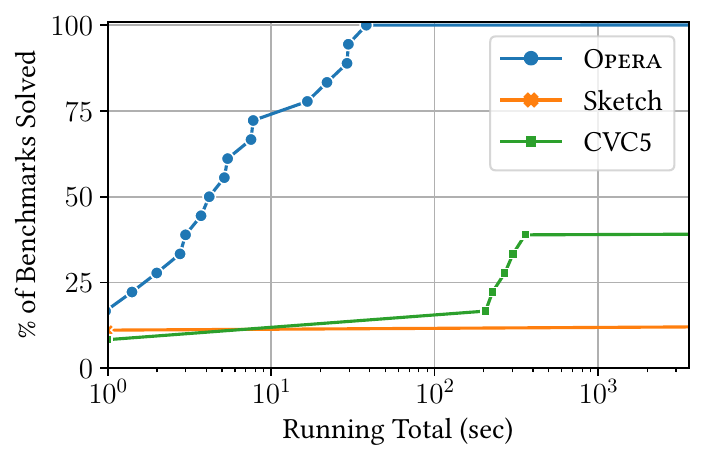}
\vspace{-0.2in}
\caption{\% of Auction Benchmarks Solved by Time}
\end{subfigure}
\vspace{-10pt}
\caption{Comparison between \tool and baselines.}\label{fig:main-cdf}
\vspace{-10pt}
\end{figure}

To answer our first two research questions, we evaluate \tool on our 51 benchmarks and compare it against two baselines.
Since there are no existing tools for generating online algorithms from their offline version, we adapt two SyGuS solvers (namely, CVC5~\cite{cvc5} and Sketch~\cite{sketch, sketch2} to our problem setting. {We chose CVC5 and Sketch among the SyGuS solvers because they support non-linear arithmetic, which is required for most of our benchmarks.})  To adapt these tools to our problem, we define the grammar of the target program to be \autoref{fig:scheme-syntax}, and we adapt the online-offline equivalence definition from \Cref{def:equiv} as the synthesis specification (using so-called \emph{oracle constraints} in SyGuS). Specifically, we assert that the synthesis result satisfies the relational function signature for a list of fixed length. We intentionally used lists of fixed size to avoid problems with the SMT solver. Finally, since SyGuS solvers require the signature of the function to be synthesized, we manually specify their signature. 
\par
\autoref{table:eval_main_res} summarizes the results of our evaluation for both \tool and the two SyGuS baselines.
In particular, 
\autoref{table:eval_main_res}  shows the percentage of benchmarks solved by each tool in the Statistics and Auction data sets, together with the average running time (in seconds) for \toolname.\footnote{The table does not report  time for the other tools because they time out on most benchmarks within the 10 min limit.} {We say that a tool solves a benchmark if it produces an online scheme that is  equivalent to the offline program, which we also verify manually.}

The key takeaway from this experiment is that \tool can solve 50 of the 51 offline programs in our benchmark suite within the 10 minute time limit. In contrast, CVC5 and Sketch solve 37\% and 14\% of the benchmarks, respectively.  We also note that average synthesis time for \tool across all benchmarks is 25.0 seconds.

To evaluate whether the  SyGuS baselines can solve more benchmarks when given a longer time limit, we also run an additional experiment with a time limit of 1 hour per task. The results of this experiment are shown in Figure~\ref{fig:main-cdf} as a cumulative distribution function (CDF) where the $x$-axis provides cumulative running time and the $y$-axis shows the percentage of benchmarks solved. As we can see, increasing the time limit does not allow any of the tools to solve additional benchmarks.

\bfpara{Qualitative Analysis for \tool.}
Of the 51 benchmarks \tool solves, we found that 41 of the synthesized schemes are the same as the manually written program. Among the 10 cases where the results differ, we found that the synthesized schemes perform the same computation but use different auxiliary parameters.
To gain more intuition about how this can happen, consider the following example: The average, $v'$, of a stream of numbers can be computed by using the sum of previously processed elements, $s$, or the previous average, $v$ as shown below:
\[
v' = (s + x) / (n + 1) \quad\quad\quad\quad v' = (v * n + x) / (n + 1).
\]
Both of them are mathematically equivalent but use different auxiliary parameters.
We note that the synthesized schemes have the same time and space complexity as the ground truth and are comparable in terms of AST size.

\begin{figure}[!t]
\centering
\begin{minted}
[
fontsize=\footnotesize,
escapeinside=||,
numbersep=5pt,
frame=lines,framerule=2pt
]
{python}
def kurtosis_online(v, m4, m3, m2, s, n, x):
    n += 1
    new_s = s + x
    delta = x - (s / n)
    delta_n = delta / n
    |\colorbox{pink}{new\_m4 = m4 + (delta * delta\_n * (n - 1) * (delta\_n**2) * (n**2 - 3 * n + 3)}|
             |\colorbox{pink}{+ 6 * delta\_n**2 * m2 - 4 * delta\_n * m3)}|
    new_m3 = m3 + delta * delta_n * (n - 1) * delta_n * (n - 2) - 3 * delta_n * m2
    new_m2 = m2 + delta * delta_n * (n - 1)
    sigma = (m2 / n) ** 0.5
    return (new_m4 / n) / (sigma**4) - 3, new_m4, new_m3, new_m2, new_s, n
\end{minted}
\vspace{-10pt}
\caption{Python implementation of online kurtosis computation.}
\label{fig:python-kurtosis}
\vspace{-10pt}
\end{figure}

\bfpara{Failure analysis.}
The only benchmark that \tool fails to solve involves computing \emph{kurtosis}, which is a measure of the tailedness of a probability distribution. \autoref{fig:python-kurtosis} shows the online algorithm for computing \emph{kurtosis} based on the method from ~\cite{sandia-report}. As we can see, the online algorithm involves a very large expression (in line 6 and also \colorbox{pink}{highlighted} in the code) that is very difficult to synthesize, so our {\sc SynthesizeExpr} procedure times out when trying to synthesize this complex expression.

\bfpara{Summary.}
Our evaluation reveals the followings answers for our first two research questions: \\

\idiotbox{RQ1}{\tool can automatically synthesize 50 out of 51 online schemes with an average synthesis time of 25.0 seconds.}

\idiotbox{RQ2}{\tool outperforms existing SyGuS solvers, synthesizing $2.6\times$ and $7.2\times$ as many tasks as CVC5 and Sketch respectively.}

\subsection{Ablation Study}
The core technical idea underlying \tool is the RFS-driven synthesis methodology, which also enables two additional optimizations used in our synthesis algorithm, namely \emph{decomposition} and the use of \emph{symbolic techniques} (namely, quantifier elimination) for expression synthesis. In this section, we evaluate the relative impact of these two ideas by considering two ablations of \tool:

\begin{enumerate}[leftmargin=*]
    \item {\bf \tool-\textsc{NoDecomp}}: This is a variant of \tool that disables compositional synthesis. In other words, rather than synthesizing a set of independent expressions, it attempts to synthesize the entire online program at once. However, it still employs the symbolic reasoning techniques that are part of the {\sc SynthesizeExpr} procedure.
    \item {\bf \tool-\textsc{NoSymbolic}}: This is a variant of \tool that replaces solver-based derivation of expressions with enumerative search. In other words, it replaces the body of {\sc SynthesizeExpr} with a call to \textsf{EnumSynthesize}.
\end{enumerate}

\begin{figure}[!t]
    \centering
    \includegraphics[width=0.6\textwidth]{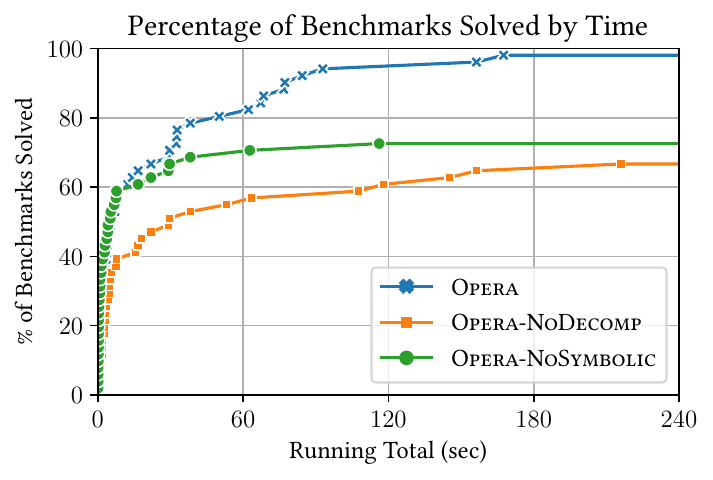}
    \vspace{-15pt}
    \caption{Comparison between \tool and its ablations.}
    \label{fig:eval_ablation}
    \vspace{-15pt}
\end{figure}

Figure~\ref{fig:eval_ablation} shows the Cumulative Distribution Function (CDF) for \tool and its ablations. As standard, the $x$-axis corresponds to the total running time, and the $y$-axis shows the percentage of benchmarks solved. As we can see, both ablations perform worse than \tool with these optimizations enabled.
In particular, the variant of \toolname without the symbolic technique solves 73\% of the benchmarks, whereas the ablation without decomposition solves 67\% of the benchmarks. For benchmarks that can be solved by both ablations, the average running time of \toolname is 10.3 seconds, whereas \tool-\textsc{NoDecomp} and \tool-\textsc{NoSymbolic} take 20.9 and 6.6 seconds respectively. As expected, decomposition has a positive impact on synthesis time regardless of the complexity of the task. In contrast, the symbolic expression synthesis technique that utilizes quantifier elimination slightly hurts performance for easy benchmarks, however, it allows significantly more benchmarks to be solved within the 10-minute time limit overall. \\

%

\idiotbox{RQ3}{Decomposition and symbolic reasoning have a significant positive impact on the performance of \tool. In particular, ablated versions of \tool without one of these optimizations solve $31\%$ and $26\%$ fewer benchmarks within the 10-minute time limit.}
%

%
\section{Limitations} \label{sec:limitations}
{
In this section, we discuss some of the main limitations of the proposed approach.

\bfpara{Limitations of problem statement.}
First, our problem statement is defined in terms of a functional IR, which means that the offline program needs to be expressible in this IR. In practice, we found that almost all offline algorithms are naturally expressed in this core functional language, and, as discussed  in \cite{HUTTON_1999}, a functional language with fold is quite expressive .
Second, our problem statement requires the synthesis result to be {\it semantically equivalent}. However, for some offline algorithms (e.g., quantile computation \cite{manku_approx_quantile}), any online algorithm that does not require remembering the entire stream necessitates approximation algorithms. We believe that synthesizing online {\it approximation} algorithms is a very interesting  direction for future work.

\bfpara{Limitations of synthesis approach.}
Our synthesis methodology relies on the assumption stated in \Cref{thm:rfs_completeness} – i.e., that it has an inductive invariant that is a conjunction of equalities. This assumption is realistic because online algorithms take additional arguments that correspond to sub-computations; thus, the inductive invariant can, in practice, always be expressed as a conjunction of equalities.   Second, as stated in \Cref{subsec:syn_online_prog}, Opera decomposes the synthesis task based on the assumption that the online program can be constructed by composing incremental computations over list expressions using operators that appear in the offline program.

\section{Related Work} \label{sec:related}
This paper is related to a long line of work on incremental computation, which attempts to only recompute those outputs which depend on changed data. Online algorithms fall under the general umbrella of incremental computation in that they compute the result one element at a time by reusing previous computations. Most of the work on incremental computation focuses on dynamic incrementalization~\cite{acar_2005, acar_2006, acar_2008, acar_2014, adapton, nominaladapton, c_memo, mapreduce_memo, reps_popl93, acar_toplas06, koenig_toplas82} by providing language support and runtime frameworks to improve running time at the cost of space. In this paper, we take a different approach by automatically synthesizing incremental online algorithms from their batch processing version. Thus, the following discussion  focuses on approaches that are more closely related to synthesis. 
\bfpara{Synthesizing incremental computation.}
There is a body of prior work on synthesizing incremental computations~\cite{annieliu, higherorder, ingress, rpai, linalg, ras_dp, sun_synthesizing_2023}. At a high level, these techniques take as input a base program $f$, a change operator $\oplus$ and attempt to generate an efficient program $f'$ that computes $f(x \oplus y)$ given $y$ and $f(x)$. Some of the existing approaches in this space are \emph{domain-specific}. For example, Shaikha et al. focus on linear algebra~\cite{linalg}, and Zhou at al. \cite{ingress} studies incremental computations related to graph processing. The technique by Pu et al. is not domain specific per se; however, their focus is on automatic derivation of dynamic programming algorithms from recurrence relations~\cite{ras_dp}.
In a similar vein, Sun et al.~\cite{sun_synthesizing_2023} studies the synthesis of efficient memoization algorithms for dynamic programming subproblems.
%

%
Among prior work on synthesizing incremental computation, the most closely related to ours is that of Liu~\cite{annieliu}, which utilizes a set of pre-defined rewrite rules to transform a base program $f$ to its incremental version $f'$. Their technique first transforms the base program to save all intermediate/auxiliary results and then tries to rewrite the program to utilize the newly introduced variables. In contrast to this rewrite-based approach, our method employs  program synthesis to solve the slightly different problem of deriving \emph{online algorithms}. More recent work by Cai et al.~\cite{higherorder} aims to statically derive incremental versions of programs written in a higher-order language. They propose a theory of changes and \emph{derivatives} and describe a type-directed method---parametrized by so-called \emph{plug-ins} for incrementalizing each type---to automatically generate a function's derivative. In contrast, our method is not type-directed and does not rely on type-specific plug-ins.
\bfpara{Related approaches in program synthesis.} This paper is also related to a long line of work on \emph{program synthesis}, which aims to generate a program from the user's specification (e.g., input-output examples or logical formula)~\cite{sketch,lambda2,synquid,flashfill,blaze,clis,semregex}. Particularly related to this work are compositional synthesis techniques that aim to decompose the original problem into independent subproblems. For example, $\lambda^2$ utilizes the semantics of functional combinators to infer input-output examples for their arguments~\cite{lambda2}, and Synquid~\cite{synquid} leverages refinement types to decompose the problem. In contrast to prior work on compositional synthesis, our method utilizes the offline program and the relational function signature to obtain \emph{completely independent} synthesis sub-tasks. \par

Another prominent aspect of our approach is the use of symbolic reasoning to derive expressions in the target program. In particular, for expression synthesis, our approach utilizes \emph{quantifier elimination} to find implicates of a certain shape. There are prior techniques that have also leveraged quantifier elimination in the context of synthesis. For example,  Comfusy \cite{comfusy1, comfusy2} and AE-VAL \cite{aeval} both apply quantifier-elimination within a deductive synthesizer to rewrite a logical specification over integer and rational arithmetic into straight-line code. The use of quantifier elimination in \tool is most closely related to the recent work of Pailoor et al.~\cite{revamp} on ADT refactoring. In that work, they utilize quantifier elimination to perform \emph{abductive reasoning} (as done previously in~\cite{explain,hola,pldi12,maximal}), and they combine abductive reasoning with search to expedite synthesis. In contrast to their approach, we use quantifier elimination to infer logical implicates of a certain shape by encoding the semantics of list combinators. Finally, recent work by Goharshady et al.~\cite{goharshady_algebro-geometric_2023} presents a promising alternative to quantifier elimination for synthesizing real valued polynomial expressions. At a high level, their approach requires a user to specify the maximum degree of the polynomial to be synthesized,  a set of variables, along with a specification, and it synthesizes a polynomial over those variables that satisfies the specification. To make synthesis scalable, they reduce the synthesis problem to an instance of quadratic programming using fundamental theorems in algebraic geometry. While this technique is specific to generating real polynomials, \tool could, in principle, apply this technique during expression synthesis if the offline expression is real-valued. 
%
%

%
Due to our use of \emph{relational function signatures} to drive online program generation, this paper is also related to \emph{relational synthesis}~\cite{relish-oopsla18,pins-pldi11,genic-pldi17,lenses-icfp19,lenses-popl18, rinard_repclass}, where the goal is to synthesize programs based on relational specifications that relate multiple programs or multiple runs of a program. For example, Relish~\cite{relish-oopsla18} leverages hierarchical finite tree automata to synthesize comparators, string encoders and decoders. Genic~\cite{genic-pldi17} and PINS~\cite{pins-pldi11} study the program inversion problem~\cite{inversion-lncs78} using symbolic extended finite transducers and path-based inductive synthesis, respectively.
There is another line of works that infer a relational specification to guide the synthesis. Mask~\cite{rinard_repclass} synthesizes replacement classes defined by the inter-class equivalence relationship.
Unlike Relish~\cite{relish-oopsla18}, Genic~\cite{genic-pldi17}, and PINS~\cite{pins-pldi11}, \tool does not have the relational specification as part of the input, so it infers an RFS as the relational specification, which is similar to Mask~\cite{rinard_repclass}. However, the relational specifications in our context are very different than those \cite{rinard_repclass}.
\tool is also related to prior work on \emph{divide-and-conquer program synthesis}~\cite{azadeh_pldi17, azadeh_pldi19, autolifter} which aims to synthesize a divide-and-conquer based procedure from a reference implementation. This is because one 
can view an online scheme as an instance of divide-and-conquer which processes the first $n-1$ elements of the stream and then joins the result with values induced by the $n^\text{th}$ element. For example, {\sc Parsynt}~\cite{azadeh_pldi17, azadeh_pldi19} transforms a single-pass algorithm into a divide-and-conquer program by lifting a sequential loop into a list homomorphism. Such a technique would not work in our context where the reference implementations can be multi-pass procedures. {\sc AutoLifter}~\cite{autolifter} is the most closely related to our approach as it removes the restriction that the implementation is single-pass and attempts to simultaneously determine the set of auxiliary variables (called \emph{aux} function in the paper) and the online program (referred to as the \emph{comb} function). Crucially, this approach first requires users to provide a relational specification between the \emph{aux} and \emph{comb} after which {\sc AutoLifter} will synthesize \emph{aux} and \emph{comb}. It does so by iteratively rewriting the specification into multiple sub-specifications that are only in terms of \emph{aux} and \emph{comb}, and then uses a CEGIS-based synthesizer to solve the sub-problems. However, unlike \tool, {\sc AutoLifter} does not exploit the syntactic structure of the offline program nor perform any symbolic reasoning to infer templates or implicates, both of which are essential to \tool's success. 
Finally, \tool is also related to prior work on \emph{recursive program synthesis} ~\cite{synduce,syrup,burst,synquid,lambda2}. Conceptually, one could  view an online scheme as a recursive function that returns the initializer in the base case and performs the online computation by combining the new input with the result of recursive calls over the processed elements. However, many of these synthesizers take as input either I/O examples or formal specifications in the form of types or logical formulas.  A more recent tool, Synduce~\cite{synduce}, utilizes the reference implementation to synthesize recursive programs; hence, it could potentially be applied to our setting, as the offline program constitutes a reference implementation. 
However, Synduce is not fully automatic as it requires the user to provide a so-called recursion skeleton. Furthermore, even when we tried to manually supply Synduce with the ground truth recursion skeleton, we were unable to get it  work on some of our simple examples, such as arithmetic mean. We conjecture that Synduce is not suitable for our setting because of the heavy use of non-linear arithmetic in these benchmarks. 


\bfpara{Program transformation and optimization techniques.}
This paper is also related to a long
line of work on program optimizations that aim to eliminate unnecessary computations.
Loop fusion is one such technique that consolidates multiple loops manipulating the same array into a single loop~\cite{loop_fusion_book}, reducing the overhead of loop control and enhancing data locality.
In the context of functional programming, lazy evaluation  allows postponing  computations until their result is actually needed~\cite{henderson_lazy_eval}. 
Work on list fusion and deforestation~\cite{coutts_stream_fusion,gill_deforestation,seidl_deforest_stop,wadler_deforest,meijer_fp} aims to eliminate intermediate data structures (e.g., lists, trees) in programs written using higher-order combinators like \texttt{map} and \texttt{fold}.
\section{Conclusion}

In this paper, we studied the problem of automatically synthesizing online streaming algorithms from their offline batch-processing version. Our method first infers a so-called \emph{relational function signature (RFS)}, which specifies the auxiliary parameters of the online program as well as how those parameters relate to computations in the offline program. Our synthesis methodology then boils down to finding an online program that is inductive relative to this RFS. Our specific synthesis algorithm uses the offline program, together with the RFS, to decompose the original problem into a set of into a set of independent sub-problems, which are solved through a combination of symbolic reasoning and search.

We have implemented the proposed approach in a tool called \tool and evaluated it on over 50 algorithms from two domains, namely, statistical computing and online auctions. Our evaluation shows that \tool can successfully solve all but one of the benchmarks and that it significantly outperforms two baselines that are adaptations of existing SyGuS solvers. Our evaluation also demonstrates the benefits of decomposition and symbolic reasoning through ablation studies.


\begin{acks}
  We thank the anonymous reviewers for their thoughtful and constructive feedback. We further acknowledge Anders Miltner for his help in formalizing the problem as well as Ben Mariano, Jingbo Wang, Celeste Barnaby, Kia Rahmani and the rest of the UToPiA group for their insightful feedback on early drafts of this paper. This work was conducted in a research group supported by \grantsponsor{GS100000001}{National Science Foundation}{http://dx.doi.org/10.13039/100000001} awards \grantnum{GS100000001}{CCF-1762299}, \grantnum{GS100000001}{CCF-1918889}, \grantnum{GS100000001}{CNS-1908304}, \grantnum{GS100000001}{CCF-1901376}, \grantnum{GS100000001}{CNS-2120696}, \grantnum{GS100000001}{CCF-2210831}, and \grantnum{GS100000001}{CCF-2319471} as well as an NSERC Discovery Grant.
\end{acks}

\bibliography{main}


\begin{thebibliography}{79}


\ifx \showCODEN    \undefined \def \showCODEN     #1{\unskip}     \fi
\ifx \showDOI      \undefined \def \showDOI       #1{#1}\fi
\ifx \showISBNx    \undefined \def \showISBNx     #1{\unskip}     \fi
\ifx \showISBNxiii \undefined \def \showISBNxiii  #1{\unskip}     \fi
\ifx \showISSN     \undefined \def \showISSN      #1{\unskip}     \fi
\ifx \showLCCN     \undefined \def \showLCCN      #1{\unskip}     \fi
\ifx \shownote     \undefined \def \shownote      #1{#1}          \fi
\ifx \showarticletitle \undefined \def \showarticletitle #1{#1}   \fi
\ifx \showURL      \undefined \def \showURL       {\relax}        \fi
\providecommand\bibfield[2]{#2}
\providecommand\bibinfo[2]{#2}
\providecommand\natexlab[1]{#1}
\providecommand\showeprint[2][]{arXiv:#2}

\bibitem[sto({[n.\,d.]})]%
        {storm}
 \bibinfo{year}{[n.\,d.]}\natexlab{}.
\newblock
\newblock
\urldef\tempurl%
\url{https://storm.apache.org/}
\showURL{%
\tempurl}


\bibitem[so_(2010)]%
        {so_post1}
 \bibinfo{year}{2010}\natexlab{}.
\newblock \bibinfo{booktitle}{}.
\newblock
\urldef\tempurl%
\url{https://web.archive.org/web/20240314171007/https://stackoverflow.com/questions/3903538/online-algorithm-for-calculating-absolute-deviation}
\showURL{%
\tempurl}
\newblock
\shownote{Accessed: 2024-03-14}.


\bibitem[so_(2013)]%
        {so_post5}
 \bibinfo{year}{2013}\natexlab{}.
\newblock \bibinfo{booktitle}{}.
\newblock
\urldef\tempurl%
\url{https://stackoverflow.com/questions/17104673/incremental-entropy-computation}
\showURL{%
\tempurl}
\newblock
\shownote{Accessed: 2024-03-14}.


\bibitem[so_(2014)]%
        {so_post4}
 \bibinfo{year}{2014}\natexlab{}.
\newblock \bibinfo{booktitle}{}.
\newblock
\urldef\tempurl%
\url{https://stackoverflow.com/questions/26191456/algorithm-for-a-running-harmonic-mean}
\showURL{%
\tempurl}
\newblock
\shownote{Accessed: 2024-03-14}.


\bibitem[so_(2018)]%
        {so_post2}
 \bibinfo{year}{2018}\natexlab{}.
\newblock \bibinfo{booktitle}{}.
\newblock
\urldef\tempurl%
\url{https://stackoverflow.com/questions/52070293/efficient-online-linear-regression-algorithm-in-python}
\showURL{%
\tempurl}
\newblock
\shownote{Accessed: 2024-03-14}.


\bibitem[so_(2023)]%
        {so_post3}
 \bibinfo{year}{2023}\natexlab{}.
\newblock \bibinfo{booktitle}{}.
\newblock
\urldef\tempurl%
\url{https://stackoverflow.com/questions/75545944/efficient-algorithm-for-online-variance-over-image-batches}
\showURL{%
\tempurl}
\newblock
\shownote{Accessed: 2024-03-14}.


\bibitem[Abeysinghe et~al\mbox{.}(2022)]%
        {rpai}
\bibfield{author}{\bibinfo{person}{Supun Abeysinghe}, \bibinfo{person}{Qiyang
  He}, {and} \bibinfo{person}{Tiark Rompf}.} \bibinfo{year}{2022}\natexlab{}.
\newblock \showarticletitle{Efficient Incrementialization of Correlated Nested
  Aggregate Queries Using Relative Partial Aggregate Indexes (RPAI)}. In
  \bibinfo{booktitle}{\emph{Proceedings of the 2022 International Conference on
  Management of Data}} (Philadelphia, PA, USA) \emph{(\bibinfo{series}{SIGMOD
  '22})}. \bibinfo{publisher}{Association for Computing Machinery},
  \bibinfo{address}{New York, NY, USA}, \bibinfo{pages}{136–149}.
\newblock
\showISBNx{9781450392495}
\urldef\tempurl%
\url{https://doi.org/10.1145/3514221.3517889}
\showDOI{\tempurl}


\bibitem[Acar et~al\mbox{.}(2006b)]%
        {acar_2006}
\bibfield{author}{\bibinfo{person}{Umut Acar}, \bibinfo{person}{Guy Blelloch},
  \bibinfo{person}{Matthias Blume}, \bibinfo{person}{Robert Harper}, {and}
  \bibinfo{person}{Kanat Tangwongsan}.} \bibinfo{year}{2006}\natexlab{b}.
\newblock \showarticletitle{A Library for Self-Adjusting Computation}.
\newblock \bibinfo{journal}{\emph{Electronic Notes in Theoretical Computer
  Science}} \bibinfo{volume}{148}, \bibinfo{number}{2} (\bibinfo{year}{2006}),
  \bibinfo{pages}{127--154}.
\newblock
\showISSN{1571-0661}
\urldef\tempurl%
\url{https://doi.org/10.1016/j.entcs.2005.11.043}
\showDOI{\tempurl}
\newblock
\shownote{Proceedings of the ACM-SIGPLAN Workshop on ML (ML 2005)}.


\bibitem[Acar(2005)]%
        {acar_2005}
\bibfield{author}{\bibinfo{person}{Umut~A. Acar}.}
  \bibinfo{year}{2005}\natexlab{}.
\newblock \emph{\bibinfo{title}{Self-adjusting computation}}.
\newblock \bibinfo{thesistype}{Ph.\,D. Dissertation}. \bibinfo{school}{School
  of Computer Science, Carnegie Mellon University}.
\newblock


\bibitem[Acar et~al\mbox{.}(2006a)]%
        {acar_toplas06}
\bibfield{author}{\bibinfo{person}{Umut~A. Acar}, \bibinfo{person}{Guy~E.
  Blelloch}, {and} \bibinfo{person}{Robert Harper}.}
  \bibinfo{year}{2006}\natexlab{a}.
\newblock \showarticletitle{Adaptive Functional Programming}.
\newblock \bibinfo{journal}{\emph{ACM Trans. Program. Lang. Syst.}}
  \bibinfo{volume}{28}, \bibinfo{number}{6} (\bibinfo{date}{nov}
  \bibinfo{year}{2006}), \bibinfo{pages}{990–1034}.
\newblock
\showISSN{0164-0925}
\urldef\tempurl%
\url{https://doi.org/10.1145/1186632.1186634}
\showDOI{\tempurl}


\bibitem[Acar and Chen(2013)]%
        {acar_2014}
\bibfield{author}{\bibinfo{person}{Umut~A. Acar} {and} \bibinfo{person}{Yan
  Chen}.} \bibinfo{year}{2013}\natexlab{}.
\newblock \showarticletitle{Streaming Big Data with Self-Adjusting
  Computation}. In \bibinfo{booktitle}{\emph{Proceedings of the 2013 Workshop
  on Data Driven Functional Programming}} (Rome, Italy)
  \emph{(\bibinfo{series}{DDFP '13})}. \bibinfo{publisher}{Association for
  Computing Machinery}, \bibinfo{address}{New York, NY, USA},
  \bibinfo{pages}{15–18}.
\newblock
\showISBNx{9781450318716}
\urldef\tempurl%
\url{https://doi.org/10.1145/2429376.2429382}
\showDOI{\tempurl}


\bibitem[Albarghouthi et~al\mbox{.}(2016)]%
        {maximal}
\bibfield{author}{\bibinfo{person}{Aws Albarghouthi}, \bibinfo{person}{Isil
  Dillig}, {and} \bibinfo{person}{Arie Gurfinkel}.}
  \bibinfo{year}{2016}\natexlab{}.
\newblock \showarticletitle{Maximal Specification Synthesis}. In
  \bibinfo{booktitle}{\emph{Proceedings of the 43rd Annual ACM SIGPLAN-SIGACT
  Symposium on Principles of Programming Languages}} (St. Petersburg, FL, USA)
  \emph{(\bibinfo{series}{POPL '16})}. \bibinfo{publisher}{Association for
  Computing Machinery}, \bibinfo{address}{New York, NY, USA},
  \bibinfo{pages}{789–801}.
\newblock
\showISBNx{9781450335492}
\urldef\tempurl%
\url{https://doi.org/10.1145/2837614.2837628}
\showDOI{\tempurl}


\bibitem[Altman and Bland(2005)]%
        {sem}
\bibfield{author}{\bibinfo{person}{Douglas~G Altman} {and}
  \bibinfo{person}{J~Martin Bland}.} \bibinfo{year}{2005}\natexlab{}.
\newblock \showarticletitle{Standard deviations and standard errors}.
\newblock \bibinfo{journal}{\emph{BMJ}} \bibinfo{volume}{331},
  \bibinfo{number}{7521} (\bibinfo{year}{2005}), \bibinfo{pages}{903}.
\newblock
\showISSN{0959-8138}
\urldef\tempurl%
\url{https://doi.org/10.1136/bmj.331.7521.903}
\showDOI{\tempurl}
\showeprint{https://www.bmj.com/content/331/7521/903.full.pdf}


\bibitem[Barbosa et~al\mbox{.}(2022)]%
        {cvc5}
\bibfield{author}{\bibinfo{person}{Haniel Barbosa}, \bibinfo{person}{Clark
  Barrett}, \bibinfo{person}{Martin Brain}, \bibinfo{person}{Gereon Kremer},
  \bibinfo{person}{Hanna Lachnitt}, \bibinfo{person}{Makai Mann},
  \bibinfo{person}{Abdalrhman Mohamed}, \bibinfo{person}{Mudathir Mohamed},
  \bibinfo{person}{Aina Niemetz}, \bibinfo{person}{Andres N{\"o}tzli},
  \bibinfo{person}{Alex Ozdemir}, \bibinfo{person}{Mathias Preiner},
  \bibinfo{person}{Andrew Reynolds}, \bibinfo{person}{Ying Sheng},
  \bibinfo{person}{Cesare Tinelli}, {and} \bibinfo{person}{Yoni Zohar}.}
  \bibinfo{year}{2022}\natexlab{}.
\newblock \showarticletitle{cvc5: A Versatile and Industrial-Strength SMT
  Solver}. In \bibinfo{booktitle}{\emph{Tools and Algorithms for the
  Construction and Analysis of Systems}},
  \bibfield{editor}{\bibinfo{person}{Dana Fisman} {and}
  \bibinfo{person}{Grigore Rosu}} (Eds.). \bibinfo{publisher}{Springer
  International Publishing}, \bibinfo{address}{Cham},
  \bibinfo{pages}{415--442}.
\newblock
\showISBNx{978-3-030-99524-9}


\bibitem[Bhatotia et~al\mbox{.}(2015)]%
        {c_memo}
\bibfield{author}{\bibinfo{person}{Pramod Bhatotia}, \bibinfo{person}{Pedro
  Fonseca}, \bibinfo{person}{Umut~A. Acar}, \bibinfo{person}{Bj\"{o}rn~B.
  Brandenburg}, {and} \bibinfo{person}{Rodrigo Rodrigues}.}
  \bibinfo{year}{2015}\natexlab{}.
\newblock \showarticletitle{IThreads: A Threading Library for Parallel
  Incremental Computation}. In \bibinfo{booktitle}{\emph{Proceedings of the
  Twentieth International Conference on Architectural Support for Programming
  Languages and Operating Systems}} (Istanbul, Turkey)
  \emph{(\bibinfo{series}{ASPLOS '15})}. \bibinfo{publisher}{Association for
  Computing Machinery}, \bibinfo{address}{New York, NY, USA},
  \bibinfo{pages}{645–659}.
\newblock
\showISBNx{9781450328357}
\urldef\tempurl%
\url{https://doi.org/10.1145/2694344.2694371}
\showDOI{\tempurl}


\bibitem[Bhatotia et~al\mbox{.}(2011)]%
        {mapreduce_memo}
\bibfield{author}{\bibinfo{person}{Pramod Bhatotia}, \bibinfo{person}{Alexander
  Wieder}, \bibinfo{person}{Rodrigo Rodrigues}, \bibinfo{person}{Umut~A. Acar},
  {and} \bibinfo{person}{Rafael Pasquin}.} \bibinfo{year}{2011}\natexlab{}.
\newblock \showarticletitle{Incoop: MapReduce for Incremental Computations}. In
  \bibinfo{booktitle}{\emph{Proceedings of the 2nd ACM Symposium on Cloud
  Computing}} (Cascais, Portugal) \emph{(\bibinfo{series}{SOCC '11})}.
  \bibinfo{publisher}{Association for Computing Machinery},
  \bibinfo{address}{New York, NY, USA}, Article \bibinfo{articleno}{7},
  \bibinfo{numpages}{14}~pages.
\newblock
\showISBNx{9781450309769}
\urldef\tempurl%
\url{https://doi.org/10.1145/2038916.2038923}
\showDOI{\tempurl}


\bibitem[Cai et~al\mbox{.}(2014)]%
        {higherorder}
\bibfield{author}{\bibinfo{person}{Yufei Cai}, \bibinfo{person}{Paolo~G.
  Giarrusso}, \bibinfo{person}{Tillmann Rendel}, {and} \bibinfo{person}{Klaus
  Ostermann}.} \bibinfo{year}{2014}\natexlab{}.
\newblock \showarticletitle{A Theory of Changes for Higher-Order Languages:
  Incrementalizing λ-Calculi by Static Differentiation}.
\newblock \bibinfo{journal}{\emph{SIGPLAN Not.}} \bibinfo{volume}{49},
  \bibinfo{number}{6} (\bibinfo{date}{jun} \bibinfo{year}{2014}),
  \bibinfo{pages}{145–155}.
\newblock
\showISSN{0362-1340}
\urldef\tempurl%
\url{https://doi.org/10.1145/2666356.2594304}
\showDOI{\tempurl}


\bibitem[Chen et~al\mbox{.}(2023)]%
        {semregex}
\bibfield{author}{\bibinfo{person}{Qiaochu Chen}, \bibinfo{person}{Arko
  Banerjee}, \bibinfo{person}{\c{C}a\u{g}atay Demiralp}, \bibinfo{person}{Greg
  Durrett}, {and} \bibinfo{person}{I\c{s}\i{}l Dillig}.}
  \bibinfo{year}{2023}\natexlab{}.
\newblock \showarticletitle{Data Extraction via Semantic Regular Expression
  Synthesis}.
\newblock \bibinfo{journal}{\emph{Proc. ACM Program. Lang.}}
  \bibinfo{volume}{7}, \bibinfo{number}{OOPSLA2}, Article
  \bibinfo{articleno}{287} (\bibinfo{date}{oct} \bibinfo{year}{2023}),
  \bibinfo{numpages}{30}~pages.
\newblock
\urldef\tempurl%
\url{https://doi.org/10.1145/3622863}
\showDOI{\tempurl}


\bibitem[Coutts et~al\mbox{.}(2007)]%
        {coutts_stream_fusion}
\bibfield{author}{\bibinfo{person}{Duncan Coutts}, \bibinfo{person}{Roman
  Leshchinskiy}, {and} \bibinfo{person}{Don Stewart}.}
  \bibinfo{year}{2007}\natexlab{}.
\newblock \showarticletitle{Stream fusion: from lists to streams to nothing at
  all}. In \bibinfo{booktitle}{\emph{Proceedings of the 12th ACM SIGPLAN
  International Conference on Functional Programming}} (Freiburg, Germany)
  \emph{(\bibinfo{series}{ICFP '07})}. \bibinfo{publisher}{Association for
  Computing Machinery}, \bibinfo{address}{New York, NY, USA},
  \bibinfo{pages}{315–326}.
\newblock
\showISBNx{9781595938152}
\urldef\tempurl%
\url{https://doi.org/10.1145/1291151.1291199}
\showDOI{\tempurl}


\bibitem[Day({[n.\,d.]})]%
        {onlinestats}
\bibfield{author}{\bibinfo{person}{Josh Day}.}
  \bibinfo{year}{[n.\,d.]}\natexlab{}.
\newblock
\newblock
\urldef\tempurl%
\url{https://github.com/joshday/OnlineStats.jl}
\showURL{%
\tempurl}


\bibitem[Dean and Ghemawat(2008)]%
        {mapreduce}
\bibfield{author}{\bibinfo{person}{Jeffrey Dean} {and} \bibinfo{person}{Sanjay
  Ghemawat}.} \bibinfo{year}{2008}\natexlab{}.
\newblock \showarticletitle{MapReduce: Simplified Data Processing on Large
  Clusters}.
\newblock \bibinfo{journal}{\emph{Commun. ACM}} \bibinfo{volume}{51},
  \bibinfo{number}{1} (\bibinfo{date}{jan} \bibinfo{year}{2008}),
  \bibinfo{pages}{107–113}.
\newblock
\showISSN{0001-0782}
\urldef\tempurl%
\url{https://doi.org/10.1145/1327452.1327492}
\showDOI{\tempurl}


\bibitem[Dijkstra(1979)]%
        {inversion-lncs78}
\bibfield{author}{\bibinfo{person}{Edsger~W. Dijkstra}.}
  \bibinfo{year}{1979}\natexlab{}.
\newblock \bibinfo{booktitle}{\emph{Program inversion}}.
\newblock \bibinfo{publisher}{Springer Berlin Heidelberg},
  \bibinfo{address}{Berlin, Heidelberg}, \bibinfo{pages}{54--57}.
\newblock
\showISBNx{978-3-540-35312-6}
\urldef\tempurl%
\url{https://doi.org/10.1007/BFb0014657}
\showDOI{\tempurl}


\bibitem[Dillig and Dillig(2013)]%
        {explain}
\bibfield{author}{\bibinfo{person}{Isil Dillig} {and} \bibinfo{person}{Thomas
  Dillig}.} \bibinfo{year}{2013}\natexlab{}.
\newblock \showarticletitle{Explain: A Tool for Performing Abductive
  Inference}. In \bibinfo{booktitle}{\emph{Proceedings of the 25th
  International Conference on Computer Aided Verification - Volume 8044}}
  (Saint Petersburg, Russia) \emph{(\bibinfo{series}{CAV 2013})}.
  \bibinfo{publisher}{Springer-Verlag}, \bibinfo{address}{Berlin, Heidelberg},
  \bibinfo{pages}{684–689}.
\newblock
\showISBNx{9783642397981}


\bibitem[Dillig et~al\mbox{.}(2012)]%
        {pldi12}
\bibfield{author}{\bibinfo{person}{Isil Dillig}, \bibinfo{person}{Thomas
  Dillig}, {and} \bibinfo{person}{Alex Aiken}.}
  \bibinfo{year}{2012}\natexlab{}.
\newblock \showarticletitle{Automated Error Diagnosis Using Abductive
  Inference}. In \bibinfo{booktitle}{\emph{Proceedings of the 33rd ACM SIGPLAN
  Conference on Programming Language Design and Implementation}} (Beijing,
  China) \emph{(\bibinfo{series}{PLDI '12})}. \bibinfo{publisher}{Association
  for Computing Machinery}, \bibinfo{address}{New York, NY, USA},
  \bibinfo{pages}{181–192}.
\newblock
\showISBNx{9781450312059}
\urldef\tempurl%
\url{https://doi.org/10.1145/2254064.2254087}
\showDOI{\tempurl}


\bibitem[Dillig et~al\mbox{.}(2013)]%
        {hola}
\bibfield{author}{\bibinfo{person}{Isil Dillig}, \bibinfo{person}{Thomas
  Dillig}, \bibinfo{person}{Boyang Li}, {and} \bibinfo{person}{Ken McMillan}.}
  \bibinfo{year}{2013}\natexlab{}.
\newblock \showarticletitle{Inductive Invariant Generation via Abductive
  Inference}. In \bibinfo{booktitle}{\emph{Proceedings of the 2013 ACM SIGPLAN
  International Conference on Object Oriented Programming Systems Languages
  \&amp; Applications}} (Indianapolis, Indiana, USA)
  \emph{(\bibinfo{series}{OOPSLA '13})}. \bibinfo{publisher}{Association for
  Computing Machinery}, \bibinfo{address}{New York, NY, USA},
  \bibinfo{pages}{443–456}.
\newblock
\showISBNx{9781450323741}
\urldef\tempurl%
\url{https://doi.org/10.1145/2509136.2509511}
\showDOI{\tempurl}


\bibitem[Farzan and Nicolet(2017)]%
        {azadeh_pldi17}
\bibfield{author}{\bibinfo{person}{Azadeh Farzan} {and} \bibinfo{person}{Victor
  Nicolet}.} \bibinfo{year}{2017}\natexlab{}.
\newblock \showarticletitle{Synthesis of divide and conquer parallelism for
  loops}. In \bibinfo{booktitle}{\emph{Proceedings of the 38th ACM SIGPLAN
  Conference on Programming Language Design and Implementation}} (Barcelona,
  Spain) \emph{(\bibinfo{series}{PLDI 2017})}. \bibinfo{publisher}{Association
  for Computing Machinery}, \bibinfo{address}{New York, NY, USA},
  \bibinfo{pages}{540–555}.
\newblock
\showISBNx{9781450349888}
\urldef\tempurl%
\url{https://doi.org/10.1145/3062341.3062355}
\showDOI{\tempurl}


\bibitem[Farzan and Nicolet(2019)]%
        {azadeh_pldi19}
\bibfield{author}{\bibinfo{person}{Azadeh Farzan} {and} \bibinfo{person}{Victor
  Nicolet}.} \bibinfo{year}{2019}\natexlab{}.
\newblock \showarticletitle{Modular divide-and-conquer parallelization of
  nested loops}. In \bibinfo{booktitle}{\emph{Proceedings of the 40th ACM
  SIGPLAN Conference on Programming Language Design and Implementation}}
  (Phoenix, AZ, USA) \emph{(\bibinfo{series}{PLDI 2019})}.
  \bibinfo{publisher}{Association for Computing Machinery},
  \bibinfo{address}{New York, NY, USA}, \bibinfo{pages}{610–624}.
\newblock
\showISBNx{9781450367127}
\urldef\tempurl%
\url{https://doi.org/10.1145/3314221.3314612}
\showDOI{\tempurl}


\bibitem[Farzan and Nicolet(2021)]%
        {synduce}
\bibfield{author}{\bibinfo{person}{Azadeh Farzan} {and} \bibinfo{person}{Victor
  Nicolet}.} \bibinfo{year}{2021}\natexlab{}.
\newblock \showarticletitle{Counterexample-Guided Partial Bounding for
  Recursive Function Synthesis}. In \bibinfo{booktitle}{\emph{Computer Aided
  Verification}}, \bibfield{editor}{\bibinfo{person}{Alexandra Silva} {and}
  \bibinfo{person}{K.~Rustan~M. Leino}} (Eds.). \bibinfo{publisher}{Springer
  International Publishing}, \bibinfo{address}{Cham},
  \bibinfo{pages}{832--855}.
\newblock
\showISBNx{978-3-030-81685-8}


\bibitem[Fedyukovich et~al\mbox{.}(2019)]%
        {aeval}
\bibfield{author}{\bibinfo{person}{Grigory Fedyukovich}, \bibinfo{person}{Arie
  Gurfinkel}, {and} \bibinfo{person}{Aarti Gupta}.}
  \bibinfo{year}{2019}\natexlab{}.
\newblock \showarticletitle{Lazy but Effective Functional Synthesis}. In
  \bibinfo{booktitle}{\emph{Verification, Model Checking, and Abstract
  Interpretation}}, \bibfield{editor}{\bibinfo{person}{Constantin Enea} {and}
  \bibinfo{person}{Ruzica Piskac}} (Eds.). \bibinfo{publisher}{Springer
  International Publishing}, \bibinfo{address}{Cham}, \bibinfo{pages}{92--113}.
\newblock
\showISBNx{978-3-030-11245-5}


\bibitem[Feser et~al\mbox{.}(2015)]%
        {lambda2}
\bibfield{author}{\bibinfo{person}{John~K. Feser}, \bibinfo{person}{Swarat
  Chaudhuri}, {and} \bibinfo{person}{Isil Dillig}.}
  \bibinfo{year}{2015}\natexlab{}.
\newblock \showarticletitle{Synthesizing Data Structure Transformations from
  Input-Output Examples}.
\newblock \bibinfo{journal}{\emph{SIGPLAN Not.}} \bibinfo{volume}{50},
  \bibinfo{number}{6} (\bibinfo{date}{jun} \bibinfo{year}{2015}),
  \bibinfo{pages}{229–239}.
\newblock
\showISSN{0362-1340}
\urldef\tempurl%
\url{https://doi.org/10.1145/2813885.2737977}
\showDOI{\tempurl}


\bibitem[Gill et~al\mbox{.}(1993)]%
        {gill_deforestation}
\bibfield{author}{\bibinfo{person}{Andrew Gill}, \bibinfo{person}{John
  Launchbury}, {and} \bibinfo{person}{Simon~L. Peyton~Jones}.}
  \bibinfo{year}{1993}\natexlab{}.
\newblock \showarticletitle{A short cut to deforestation}. In
  \bibinfo{booktitle}{\emph{Proceedings of the Conference on Functional
  Programming Languages and Computer Architecture}} (Copenhagen, Denmark)
  \emph{(\bibinfo{series}{FPCA '93})}. \bibinfo{publisher}{Association for
  Computing Machinery}, \bibinfo{address}{New York, NY, USA},
  \bibinfo{pages}{223–232}.
\newblock
\showISBNx{089791595X}
\urldef\tempurl%
\url{https://doi.org/10.1145/165180.165214}
\showDOI{\tempurl}


\bibitem[Goharshady et~al\mbox{.}(2023)]%
        {goharshady_algebro-geometric_2023}
\bibfield{author}{\bibinfo{person}{Amir~Kafshdar Goharshady},
  \bibinfo{person}{S. Hitarth}, \bibinfo{person}{Fatemeh Mohammadi}, {and}
  \bibinfo{person}{Harshit~Jitendra Motwani}.} \bibinfo{year}{2023}\natexlab{}.
\newblock \showarticletitle{Algebro-geometric {Algorithms} for
  {Template}-{Based} {Synthesis} of {Polynomial} {Programs}}.
\newblock \bibinfo{journal}{\emph{Proceedings of the ACM on Programming
  Languages}} \bibinfo{volume}{7}, \bibinfo{number}{OOPSLA1}
  (\bibinfo{date}{April} \bibinfo{year}{2023}),
  \bibinfo{pages}{100:727--100:756}.
\newblock
\urldef\tempurl%
\url{https://doi.org/10.1145/3586052}
\showDOI{\tempurl}


\bibitem[Gong et~al\mbox{.}(2021)]%
        {ingress}
\bibfield{author}{\bibinfo{person}{Shufeng Gong}, \bibinfo{person}{Chao Tian},
  \bibinfo{person}{Qiang Yin}, \bibinfo{person}{Wenyuan Yu},
  \bibinfo{person}{Yanfeng Zhang}, \bibinfo{person}{Liang Geng},
  \bibinfo{person}{Song Yu}, \bibinfo{person}{Ge Yu}, {and}
  \bibinfo{person}{Jingren Zhou}.} \bibinfo{year}{2021}\natexlab{}.
\newblock \showarticletitle{Automating Incremental Graph Processing with
  Flexible Memoization}.
\newblock \bibinfo{journal}{\emph{Proc. VLDB Endow.}} \bibinfo{volume}{14},
  \bibinfo{number}{9} (\bibinfo{date}{may} \bibinfo{year}{2021}),
  \bibinfo{pages}{1613–1625}.
\newblock
\showISSN{2150-8097}
\urldef\tempurl%
\url{https://doi.org/10.14778/3461535.3461550}
\showDOI{\tempurl}


\bibitem[Griss(1975)]%
        {reduce}
\bibfield{author}{\bibinfo{person}{Martin~L. Griss}.}
  \bibinfo{year}{1975}\natexlab{}.
\newblock \showarticletitle{The REDUCE System for Computer Algebra}. In
  \bibinfo{booktitle}{\emph{Proceedings of the 1975 Annual Conference}}
  \emph{(\bibinfo{series}{ACM '75})}. \bibinfo{publisher}{Association for
  Computing Machinery}, \bibinfo{address}{New York, NY, USA},
  \bibinfo{pages}{261–262}.
\newblock
\showISBNx{9781450374811}
\urldef\tempurl%
\url{https://doi.org/10.1145/800181.810335}
\showDOI{\tempurl}


\bibitem[Gulwani(2011)]%
        {flashfill}
\bibfield{author}{\bibinfo{person}{Sumit Gulwani}.}
  \bibinfo{year}{2011}\natexlab{}.
\newblock \showarticletitle{Automating String Processing in Spreadsheets Using
  Input-Output Examples}.
\newblock \bibinfo{journal}{\emph{SIGPLAN Not.}} \bibinfo{volume}{46},
  \bibinfo{number}{1} (\bibinfo{date}{jan} \bibinfo{year}{2011}),
  \bibinfo{pages}{317–330}.
\newblock
\showISSN{0362-1340}
\urldef\tempurl%
\url{https://doi.org/10.1145/1925844.1926423}
\showDOI{\tempurl}


\bibitem[Hammer and Acar(2008)]%
        {acar_2008}
\bibfield{author}{\bibinfo{person}{Matthew~A. Hammer} {and}
  \bibinfo{person}{Umut~A. Acar}.} \bibinfo{year}{2008}\natexlab{}.
\newblock \showarticletitle{Memory management for self-adjusting computation}.
  In \bibinfo{booktitle}{\emph{Proceedings of the 7th International Symposium
  on Memory Management, {ISMM} 2008, Tucson, AZ, USA, June 7-8, 2008}},
  \bibfield{editor}{\bibinfo{person}{Richard~E. Jones} {and}
  \bibinfo{person}{Stephen~M. Blackburn}} (Eds.). \bibinfo{publisher}{{ACM}},
  \bibinfo{pages}{51--60}.
\newblock
\urldef\tempurl%
\url{https://doi.org/10.1145/1375634.1375642}
\showDOI{\tempurl}


\bibitem[Hammer et~al\mbox{.}(2015)]%
        {nominaladapton}
\bibfield{author}{\bibinfo{person}{Matthew~A. Hammer}, \bibinfo{person}{Jana
  Dunfield}, \bibinfo{person}{Kyle Headley}, \bibinfo{person}{Nicholas Labich},
  \bibinfo{person}{Jeffrey~S. Foster}, \bibinfo{person}{Michael Hicks}, {and}
  \bibinfo{person}{David Van~Horn}.} \bibinfo{year}{2015}\natexlab{}.
\newblock \showarticletitle{Incremental Computation with Names}.
\newblock \bibinfo{journal}{\emph{SIGPLAN Not.}} \bibinfo{volume}{50},
  \bibinfo{number}{10} (\bibinfo{date}{oct} \bibinfo{year}{2015}),
  \bibinfo{pages}{748–766}.
\newblock
\showISSN{0362-1340}
\urldef\tempurl%
\url{https://doi.org/10.1145/2858965.2814305}
\showDOI{\tempurl}


\bibitem[Hammer et~al\mbox{.}(2014)]%
        {adapton}
\bibfield{author}{\bibinfo{person}{Matthew~A. Hammer},
  \bibinfo{person}{Khoo~Yit Phang}, \bibinfo{person}{Michael Hicks}, {and}
  \bibinfo{person}{Jeffrey~S. Foster}.} \bibinfo{year}{2014}\natexlab{}.
\newblock \showarticletitle{Adapton: Composable, Demand-Driven Incremental
  Computation}. In \bibinfo{booktitle}{\emph{Proceedings of the 35th ACM
  SIGPLAN Conference on Programming Language Design and Implementation}}
  (Edinburgh, United Kingdom) \emph{(\bibinfo{series}{PLDI '14})}.
  \bibinfo{publisher}{Association for Computing Machinery},
  \bibinfo{address}{New York, NY, USA}, \bibinfo{pages}{156–166}.
\newblock
\showISBNx{9781450327848}
\urldef\tempurl%
\url{https://doi.org/10.1145/2594291.2594324}
\showDOI{\tempurl}


\bibitem[Henderson and Morris(1976)]%
        {henderson_lazy_eval}
\bibfield{author}{\bibinfo{person}{Peter Henderson} {and}
  \bibinfo{person}{James~H. Morris}.} \bibinfo{year}{1976}\natexlab{}.
\newblock \showarticletitle{A lazy evaluator}. In
  \bibinfo{booktitle}{\emph{Proceedings of the 3rd ACM SIGACT-SIGPLAN Symposium
  on Principles on Programming Languages}} (Atlanta, Georgia)
  \emph{(\bibinfo{series}{POPL '76})}. \bibinfo{publisher}{Association for
  Computing Machinery}, \bibinfo{address}{New York, NY, USA},
  \bibinfo{pages}{95–103}.
\newblock
\showISBNx{9781450374774}
\urldef\tempurl%
\url{https://doi.org/10.1145/800168.811543}
\showDOI{\tempurl}


\bibitem[Hu and D'Antoni(2017)]%
        {genic-pldi17}
\bibfield{author}{\bibinfo{person}{Qinheping Hu} {and} \bibinfo{person}{Loris
  D'Antoni}.} \bibinfo{year}{2017}\natexlab{}.
\newblock \showarticletitle{Automatic Program Inversion Using Symbolic
  Transducers}.
\newblock \bibinfo{journal}{\emph{SIGPLAN Not.}} \bibinfo{volume}{52},
  \bibinfo{number}{6} (\bibinfo{date}{jun} \bibinfo{year}{2017}),
  \bibinfo{pages}{376–389}.
\newblock
\showISSN{0362-1340}
\urldef\tempurl%
\url{https://doi.org/10.1145/3140587.3062345}
\showDOI{\tempurl}


\bibitem[Hutton(1999)]%
        {HUTTON_1999}
\bibfield{author}{\bibinfo{person}{Graham Hutton}.}
  \bibinfo{year}{1999}\natexlab{}.
\newblock \showarticletitle{A tutorial on the universality and expressiveness
  of fold}.
\newblock \bibinfo{journal}{\emph{Journal of Functional Programming}}
  \bibinfo{volume}{9}, \bibinfo{number}{4} (\bibinfo{year}{1999}),
  \bibinfo{pages}{355–372}.
\newblock
\urldef\tempurl%
\url{https://doi.org/10.1017/S0956796899003500}
\showDOI{\tempurl}


\bibitem[Ji et~al\mbox{.}(2024)]%
        {autolifter}
\bibfield{author}{\bibinfo{person}{Ruyi Ji}, \bibinfo{person}{Yuwei Zhao},
  \bibinfo{person}{Yingfei Xiong}, \bibinfo{person}{Di Wang},
  \bibinfo{person}{Lu Zhang}, {and} \bibinfo{person}{Zhenjiang Hu}.}
  \bibinfo{year}{2024}\natexlab{}.
\newblock \showarticletitle{Decomposition-Based Synthesis for Applying
  Divide-and-Conquer-Like Algorithmic Paradigms}.
\newblock \bibinfo{journal}{\emph{ACM Trans. Program. Lang. Syst.}}
  (\bibinfo{date}{feb} \bibinfo{year}{2024}).
\newblock
\showISSN{0164-0925}
\urldef\tempurl%
\url{https://doi.org/10.1145/3648440}
\showDOI{\tempurl}
\newblock
\shownote{Just Accepted}.


\bibitem[Katsifodimos and Schelter(2016)]%
        {flink}
\bibfield{author}{\bibinfo{person}{Asterios Katsifodimos} {and}
  \bibinfo{person}{Sebastian Schelter}.} \bibinfo{year}{2016}\natexlab{}.
\newblock \showarticletitle{Apache Flink: Stream Analytics at Scale}. In
  \bibinfo{booktitle}{\emph{2016 IEEE International Conference on Cloud
  Engineering Workshop (IC2EW)}}. \bibinfo{pages}{193--193}.
\newblock
\urldef\tempurl%
\url{https://doi.org/10.1109/IC2EW.2016.56}
\showDOI{\tempurl}


\bibitem[Kennedy and Allen(2001)]%
        {loop_fusion_book}
\bibfield{author}{\bibinfo{person}{Ken Kennedy} {and} \bibinfo{person}{John~R.
  Allen}.} \bibinfo{year}{2001}\natexlab{}.
\newblock \bibinfo{booktitle}{\emph{Optimizing compilers for modern
  architectures: a dependence-based approach}}.
\newblock \bibinfo{publisher}{Morgan Kaufmann Publishers Inc.},
  \bibinfo{address}{San Francisco, CA, USA}.
\newblock
\showISBNx{1558602860}


\bibitem[Kreps et~al\mbox{.}(2011)]%
        {kafka}
\bibfield{author}{\bibinfo{person}{Jay Kreps}, \bibinfo{person}{Neha Narkhede},
  \bibinfo{person}{Jun Rao}, {et~al\mbox{.}}} \bibinfo{year}{2011}\natexlab{}.
\newblock \showarticletitle{Kafka: A distributed messaging system for log
  processing}. In \bibinfo{booktitle}{\emph{Proceedings of the NetDB}},
  Vol.~\bibinfo{volume}{11}. Athens, Greece, \bibinfo{pages}{1--7}.
\newblock


\bibitem[Kuncak et~al\mbox{.}(2010a)]%
        {comfusy1}
\bibfield{author}{\bibinfo{person}{Viktor Kuncak}, \bibinfo{person}{Mika\"{e}l
  Mayer}, \bibinfo{person}{Ruzica Piskac}, {and} \bibinfo{person}{Philippe
  Suter}.} \bibinfo{year}{2010}\natexlab{a}.
\newblock \showarticletitle{Comfusy: A Tool for Complete Functional Synthesis}.
  In \bibinfo{booktitle}{\emph{Proceedings of the 22nd International Conference
  on Computer Aided Verification}} (Edinburgh, UK)
  \emph{(\bibinfo{series}{CAV'10})}. \bibinfo{publisher}{Springer-Verlag},
  \bibinfo{address}{Berlin, Heidelberg}, \bibinfo{pages}{430–433}.
\newblock
\showISBNx{364214294X}
\urldef\tempurl%
\url{https://doi.org/10.1007/978-3-642-14295-6_38}
\showDOI{\tempurl}


\bibitem[Kuncak et~al\mbox{.}(2010b)]%
        {comfusy2}
\bibfield{author}{\bibinfo{person}{Viktor Kuncak}, \bibinfo{person}{Mika\"{e}l
  Mayer}, \bibinfo{person}{Ruzica Piskac}, {and} \bibinfo{person}{Philippe
  Suter}.} \bibinfo{year}{2010}\natexlab{b}.
\newblock \showarticletitle{Complete Functional Synthesis}. In
  \bibinfo{booktitle}{\emph{Proceedings of the 31st ACM SIGPLAN Conference on
  Programming Language Design and Implementation}} (Toronto, Ontario, Canada)
  \emph{(\bibinfo{series}{PLDI '10})}. \bibinfo{publisher}{Association for
  Computing Machinery}, \bibinfo{address}{New York, NY, USA},
  \bibinfo{pages}{316–329}.
\newblock
\showISBNx{9781450300193}
\urldef\tempurl%
\url{https://doi.org/10.1145/1806596.1806632}
\showDOI{\tempurl}


\bibitem[Liu(2000)]%
        {annieliu}
\bibfield{author}{\bibinfo{person}{Yanhong~A. Liu}.}
  \bibinfo{year}{2000}\natexlab{}.
\newblock \showarticletitle{Efficiency by Incrementalization: An Introduction}.
\newblock \bibinfo{journal}{\emph{Higher Order Symbol. Comput.}}
  \bibinfo{volume}{13}, \bibinfo{number}{4} (\bibinfo{date}{dec}
  \bibinfo{year}{2000}), \bibinfo{pages}{289–313}.
\newblock
\showISSN{1388-3690}
\urldef\tempurl%
\url{https://doi.org/10.1023/A:1026547031739}
\showDOI{\tempurl}


\bibitem[Manku et~al\mbox{.}(1998)]%
        {manku_approx_quantile}
\bibfield{author}{\bibinfo{person}{Gurmeet~Singh Manku},
  \bibinfo{person}{Sridhar Rajagopalan}, {and} \bibinfo{person}{Bruce~G.
  Lindsay}.} \bibinfo{year}{1998}\natexlab{}.
\newblock \showarticletitle{Approximate medians and other quantiles in one pass
  and with limited memory}. In \bibinfo{booktitle}{\emph{Proceedings of the
  1998 ACM SIGMOD International Conference on Management of Data}} (Seattle,
  Washington, USA) \emph{(\bibinfo{series}{SIGMOD '98})}.
  \bibinfo{publisher}{Association for Computing Machinery},
  \bibinfo{address}{New York, NY, USA}, \bibinfo{pages}{426–435}.
\newblock
\showISBNx{0897919955}
\urldef\tempurl%
\url{https://doi.org/10.1145/276304.276342}
\showDOI{\tempurl}


\bibitem[Mariano et~al\mbox{.}(2022)]%
        {ben-oopsla22}
\bibfield{author}{\bibinfo{person}{Benjamin Mariano}, \bibinfo{person}{Yanju
  Chen}, \bibinfo{person}{Yu Feng}, \bibinfo{person}{Greg Durrett}, {and}
  \bibinfo{person}{I\c{s}il Dillig}.} \bibinfo{year}{2022}\natexlab{}.
\newblock \showarticletitle{Automated Transpilation of Imperative to Functional
  Code Using Neural-Guided Program Synthesis}.
\newblock \bibinfo{journal}{\emph{Proc. ACM Program. Lang.}}
  \bibinfo{volume}{6}, \bibinfo{number}{OOPSLA1}, Article
  \bibinfo{articleno}{71} (\bibinfo{date}{April} \bibinfo{year}{2022}),
  \bibinfo{numpages}{27}~pages.
\newblock
\urldef\tempurl%
\url{https://doi.org/10.1145/3527315}
\showDOI{\tempurl}


\bibitem[Meijer et~al\mbox{.}(1991)]%
        {meijer_fp}
\bibfield{author}{\bibinfo{person}{Erik Meijer}, \bibinfo{person}{Maarten
  Fokkinga}, {and} \bibinfo{person}{Ross Paterson}.}
  \bibinfo{year}{1991}\natexlab{}.
\newblock \showarticletitle{Functional programming with bananas, lenses,
  envelopes and barbed wire}. In \bibinfo{booktitle}{\emph{Functional
  Programming Languages and Computer Architecture}},
  \bibfield{editor}{\bibinfo{person}{John Hughes}} (Ed.).
  \bibinfo{publisher}{Springer Berlin Heidelberg}, \bibinfo{address}{Berlin,
  Heidelberg}, \bibinfo{pages}{124--144}.
\newblock
\showISBNx{978-3-540-47599-6}


\bibitem[Miltner et~al\mbox{.}(2017)]%
        {lenses-popl18}
\bibfield{author}{\bibinfo{person}{Anders Miltner}, \bibinfo{person}{Kathleen
  Fisher}, \bibinfo{person}{Benjamin~C. Pierce}, \bibinfo{person}{David
  Walker}, {and} \bibinfo{person}{Steve Zdancewic}.}
  \bibinfo{year}{2017}\natexlab{}.
\newblock \showarticletitle{Synthesizing Bijective Lenses}.
\newblock \bibinfo{journal}{\emph{Proc. ACM Program. Lang.}}
  \bibinfo{volume}{2}, \bibinfo{number}{POPL}, Article \bibinfo{articleno}{1}
  (\bibinfo{date}{dec} \bibinfo{year}{2017}), \bibinfo{numpages}{30}~pages.
\newblock
\urldef\tempurl%
\url{https://doi.org/10.1145/3158089}
\showDOI{\tempurl}


\bibitem[Miltner et~al\mbox{.}(2019)]%
        {lenses-icfp19}
\bibfield{author}{\bibinfo{person}{Anders Miltner}, \bibinfo{person}{Solomon
  Maina}, \bibinfo{person}{Kathleen Fisher}, \bibinfo{person}{Benjamin~C.
  Pierce}, \bibinfo{person}{David Walker}, {and} \bibinfo{person}{Steve
  Zdancewic}.} \bibinfo{year}{2019}\natexlab{}.
\newblock \showarticletitle{Synthesizing Symmetric Lenses}.
\newblock \bibinfo{journal}{\emph{Proc. ACM Program. Lang.}}
  \bibinfo{volume}{3}, \bibinfo{number}{ICFP}, Article \bibinfo{articleno}{95}
  (\bibinfo{date}{jul} \bibinfo{year}{2019}), \bibinfo{numpages}{28}~pages.
\newblock
\urldef\tempurl%
\url{https://doi.org/10.1145/3341699}
\showDOI{\tempurl}


\bibitem[Miltner et~al\mbox{.}(2022)]%
        {burst}
\bibfield{author}{\bibinfo{person}{Anders Miltner},
  \bibinfo{person}{Adrian~Trejo Nu\~{n}ez}, \bibinfo{person}{Ana Brendel},
  \bibinfo{person}{Swarat Chaudhuri}, {and} \bibinfo{person}{Isil Dillig}.}
  \bibinfo{year}{2022}\natexlab{}.
\newblock \showarticletitle{Bottom-up Synthesis of Recursive Functional
  Programs Using Angelic Execution}.
\newblock \bibinfo{journal}{\emph{Proc. ACM Program. Lang.}}
  \bibinfo{volume}{6}, \bibinfo{number}{POPL}, Article \bibinfo{articleno}{21}
  (\bibinfo{date}{jan} \bibinfo{year}{2022}), \bibinfo{numpages}{29}~pages.
\newblock
\urldef\tempurl%
\url{https://doi.org/10.1145/3498682}
\showDOI{\tempurl}


\bibitem[Noghabi et~al\mbox{.}(2017)]%
        {samza}
\bibfield{author}{\bibinfo{person}{Shadi~A. Noghabi}, \bibinfo{person}{Kartik
  Paramasivam}, \bibinfo{person}{Yi Pan}, \bibinfo{person}{Navina Ramesh},
  \bibinfo{person}{Jon Bringhurst}, \bibinfo{person}{Indranil Gupta}, {and}
  \bibinfo{person}{Roy~H. Campbell}.} \bibinfo{year}{2017}\natexlab{}.
\newblock \showarticletitle{Samza: Stateful Scalable Stream Processing at
  LinkedIn}.
\newblock \bibinfo{journal}{\emph{Proc. VLDB Endow.}} \bibinfo{volume}{10},
  \bibinfo{number}{12} (\bibinfo{date}{aug} \bibinfo{year}{2017}),
  \bibinfo{pages}{1634–1645}.
\newblock
\showISSN{2150-8097}
\urldef\tempurl%
\url{https://doi.org/10.14778/3137765.3137770}
\showDOI{\tempurl}


\bibitem[Paige and Koenig(1982)]%
        {koenig_toplas82}
\bibfield{author}{\bibinfo{person}{Robert Paige} {and} \bibinfo{person}{Shaye
  Koenig}.} \bibinfo{year}{1982}\natexlab{}.
\newblock \showarticletitle{Finite Differencing of Computable Expressions}.
\newblock \bibinfo{journal}{\emph{ACM Trans. Program. Lang. Syst.}}
  \bibinfo{volume}{4}, \bibinfo{number}{3} (\bibinfo{date}{jul}
  \bibinfo{year}{1982}), \bibinfo{pages}{402–454}.
\newblock
\showISSN{0164-0925}
\urldef\tempurl%
\url{https://doi.org/10.1145/357172.357177}
\showDOI{\tempurl}


\bibitem[Pailoor et~al\mbox{.}(2024)]%
        {revamp}
\bibfield{author}{\bibinfo{person}{Shankara Pailoor}, \bibinfo{person}{Yuepeng
  Wang}, {and} \bibinfo{person}{Işıl Dillig}.}
  \bibinfo{year}{2024}\natexlab{}.
\newblock \showarticletitle{Semantic Code Refactoring for Abstract Data Types}.
\newblock \bibinfo{journal}{\emph{Proc. ACM Program. Lang.}}
  \bibinfo{volume}{8}, \bibinfo{number}{POPL}, Article \bibinfo{articleno}{28}
  (\bibinfo{date}{January} \bibinfo{year}{2024}), \bibinfo{numpages}{32}~pages.
\newblock
\urldef\tempurl%
\url{https://doi.org/10.1145/3632870}
\showDOI{\tempurl}


\bibitem[Pebay(2008)]%
        {sandia-report}
\bibfield{author}{\bibinfo{person}{Philippe~Pierre Pebay}.}
  \bibinfo{year}{2008}\natexlab{}.
\newblock \showarticletitle{Formulas for robust, one-pass parallel computation
  of covariances and arbitrary-order statistical moments}.
\newblock  (\bibinfo{date}{01} \bibinfo{year}{2008}).
\newblock
\urldef\tempurl%
\url{https://doi.org/10.2172/1028931}
\showDOI{\tempurl}


\bibitem[Polikarpova et~al\mbox{.}(2016)]%
        {synquid}
\bibfield{author}{\bibinfo{person}{Nadia Polikarpova}, \bibinfo{person}{Ivan
  Kuraj}, {and} \bibinfo{person}{Armando Solar-Lezama}.}
  \bibinfo{year}{2016}\natexlab{}.
\newblock \showarticletitle{Program Synthesis from Polymorphic Refinement
  Types}.
\newblock \bibinfo{journal}{\emph{SIGPLAN Not.}} \bibinfo{volume}{51},
  \bibinfo{number}{6} (\bibinfo{date}{jun} \bibinfo{year}{2016}),
  \bibinfo{pages}{522–538}.
\newblock
\showISSN{0362-1340}
\urldef\tempurl%
\url{https://doi.org/10.1145/2980983.2908093}
\showDOI{\tempurl}


\bibitem[Pu et~al\mbox{.}(2011)]%
        {ras_dp}
\bibfield{author}{\bibinfo{person}{Yewen Pu}, \bibinfo{person}{Rastislav
  Bodik}, {and} \bibinfo{person}{Saurabh Srivastava}.}
  \bibinfo{year}{2011}\natexlab{}.
\newblock \showarticletitle{Synthesis of First-Order Dynamic Programming
  Algorithms}. In \bibinfo{booktitle}{\emph{Proceedings of the 2011 ACM
  International Conference on Object Oriented Programming Systems Languages and
  Applications}} (Portland, Oregon, USA) \emph{(\bibinfo{series}{OOPSLA '11})}.
  \bibinfo{publisher}{Association for Computing Machinery},
  \bibinfo{address}{New York, NY, USA}, \bibinfo{pages}{83–98}.
\newblock
\showISBNx{9781450309400}
\urldef\tempurl%
\url{https://doi.org/10.1145/2048066.2048076}
\showDOI{\tempurl}


\bibitem[Ramalingam and Reps(1993)]%
        {reps_popl93}
\bibfield{author}{\bibinfo{person}{G. Ramalingam} {and} \bibinfo{person}{Thomas
  Reps}.} \bibinfo{year}{1993}\natexlab{}.
\newblock \showarticletitle{A Categorized Bibliography on Incremental
  Computation}. In \bibinfo{booktitle}{\emph{Proceedings of the 20th ACM
  SIGPLAN-SIGACT Symposium on Principles of Programming Languages}}
  (Charleston, South Carolina, USA) \emph{(\bibinfo{series}{POPL '93})}.
  \bibinfo{publisher}{Association for Computing Machinery},
  \bibinfo{address}{New York, NY, USA}, \bibinfo{pages}{502–510}.
\newblock
\showISBNx{0897915607}
\urldef\tempurl%
\url{https://doi.org/10.1145/158511.158710}
\showDOI{\tempurl}


\bibitem[Samak et~al\mbox{.}(2019)]%
        {rinard_repclass}
\bibfield{author}{\bibinfo{person}{Malavika Samak}, \bibinfo{person}{Deokhwan
  Kim}, {and} \bibinfo{person}{Martin~C. Rinard}.}
  \bibinfo{year}{2019}\natexlab{}.
\newblock \showarticletitle{Synthesizing Replacement Classes}.
\newblock \bibinfo{journal}{\emph{Proc. ACM Program. Lang.}}
  \bibinfo{volume}{4}, \bibinfo{number}{POPL}, Article \bibinfo{articleno}{52}
  (\bibinfo{date}{dec} \bibinfo{year}{2019}), \bibinfo{numpages}{33}~pages.
\newblock
\urldef\tempurl%
\url{https://doi.org/10.1145/3371120}
\showDOI{\tempurl}


\bibitem[Seidl and S\o{}rensen(1997)]%
        {seidl_deforest_stop}
\bibfield{author}{\bibinfo{person}{H. Seidl} {and} \bibinfo{person}{M.~H.
  S\o{}rensen}.} \bibinfo{year}{1997}\natexlab{}.
\newblock \showarticletitle{Constraints to stop higher-order deforestation}. In
  \bibinfo{booktitle}{\emph{Proceedings of the 24th ACM SIGPLAN-SIGACT
  Symposium on Principles of Programming Languages}} (Paris, France)
  \emph{(\bibinfo{series}{POPL '97})}. \bibinfo{publisher}{Association for
  Computing Machinery}, \bibinfo{address}{New York, NY, USA},
  \bibinfo{pages}{400–413}.
\newblock
\showISBNx{0897918533}
\urldef\tempurl%
\url{https://doi.org/10.1145/263699.263758}
\showDOI{\tempurl}


\bibitem[Shaikhha et~al\mbox{.}(2020)]%
        {linalg}
\bibfield{author}{\bibinfo{person}{Amir Shaikhha}, \bibinfo{person}{Mohammed
  Elseidy}, \bibinfo{person}{Stephan Mihaila}, \bibinfo{person}{Daniel Espino},
  {and} \bibinfo{person}{Christoph Koch}.} \bibinfo{year}{2020}\natexlab{}.
\newblock \showarticletitle{Synthesis of Incremental Linear Algebra Programs}.
\newblock \bibinfo{journal}{\emph{ACM Trans. Database Syst.}}
  \bibinfo{volume}{45}, \bibinfo{number}{3}, Article \bibinfo{articleno}{12}
  (\bibinfo{date}{aug} \bibinfo{year}{2020}), \bibinfo{numpages}{44}~pages.
\newblock
\showISSN{0362-5915}
\urldef\tempurl%
\url{https://doi.org/10.1145/3385398}
\showDOI{\tempurl}


\bibitem[Solar{-}Lezama et~al\mbox{.}(2008)]%
        {sketch2}
\bibfield{author}{\bibinfo{person}{Armando Solar{-}Lezama},
  \bibinfo{person}{Christopher~Grant Jones}, {and} \bibinfo{person}{Rastislav
  Bod{\'{\i}}k}.} \bibinfo{year}{2008}\natexlab{}.
\newblock \showarticletitle{Sketching concurrent data structures}. In
  \bibinfo{booktitle}{\emph{Proceedings of the {ACM} {SIGPLAN} 2008 Conference
  on Programming Language Design and Implementation ({PLDI})}}.
  \bibinfo{publisher}{{ACM}}, \bibinfo{pages}{136--148}.
\newblock
\urldef\tempurl%
\url{https://doi.org/10.1145/1375581.1375599}
\showDOI{\tempurl}


\bibitem[Solar{-}Lezama et~al\mbox{.}(2006)]%
        {sketch}
\bibfield{author}{\bibinfo{person}{Armando Solar{-}Lezama},
  \bibinfo{person}{Liviu Tancau}, \bibinfo{person}{Rastislav Bod{\'{\i}}k},
  \bibinfo{person}{Sanjit~A. Seshia}, {and} \bibinfo{person}{Vijay~A.
  Saraswat}.} \bibinfo{year}{2006}\natexlab{}.
\newblock \showarticletitle{Combinatorial sketching for finite programs}. In
  \bibinfo{booktitle}{\emph{Proceedings of the 12th International Conference on
  Architectural Support for Programming Languages and Operating Systems
  ({ASPLOS})}}. \bibinfo{publisher}{{ACM}}, \bibinfo{pages}{404--415}.
\newblock
\urldef\tempurl%
\url{https://doi.org/10.1145/1168857.1168907}
\showDOI{\tempurl}


\bibitem[Srivastava et~al\mbox{.}(2011)]%
        {pins-pldi11}
\bibfield{author}{\bibinfo{person}{Saurabh Srivastava}, \bibinfo{person}{Sumit
  Gulwani}, \bibinfo{person}{Swarat Chaudhuri}, {and}
  \bibinfo{person}{Jeffrey~S. Foster}.} \bibinfo{year}{2011}\natexlab{}.
\newblock \showarticletitle{Path-Based Inductive Synthesis for Program
  Inversion}. In \bibinfo{booktitle}{\emph{Proceedings of the 32nd ACM SIGPLAN
  Conference on Programming Language Design and Implementation}} (San Jose,
  California, USA) \emph{(\bibinfo{series}{PLDI '11})}.
  \bibinfo{publisher}{Association for Computing Machinery},
  \bibinfo{address}{New York, NY, USA}, \bibinfo{pages}{492–503}.
\newblock
\showISBNx{9781450306638}
\urldef\tempurl%
\url{https://doi.org/10.1145/1993498.1993557}
\showDOI{\tempurl}


\bibitem[Sun et~al\mbox{.}(2023)]%
        {sun_synthesizing_2023}
\bibfield{author}{\bibinfo{person}{Yican Sun}, \bibinfo{person}{Xuanyu Peng},
  {and} \bibinfo{person}{Yingfei Xiong}.} \bibinfo{year}{2023}\natexlab{}.
\newblock \showarticletitle{Synthesizing {Efficient} {Memoization}
  {Algorithms}}.
\newblock \bibinfo{journal}{\emph{Proceedings of the ACM on Programming
  Languages}} \bibinfo{volume}{7}, \bibinfo{number}{OOPSLA2}
  (\bibinfo{date}{Oct.} \bibinfo{year}{2023}),
  \bibinfo{pages}{225:89--225:115}.
\newblock
\urldef\tempurl%
\url{https://doi.org/10.1145/3622800}
\showDOI{\tempurl}


\bibitem[Tucker et~al\mbox{.}(2010)]%
        {nexmark}
\bibfield{author}{\bibinfo{person}{Peter~A. Tucker}, \bibinfo{person}{Kristin
  Tufte}, \bibinfo{person}{Vassilis Papadimos}, {and} \bibinfo{person}{David
  Maier}.} \bibinfo{year}{2010}\natexlab{}.
\newblock \showarticletitle{NEXMark – A Benchmark for Queries over Data
  Streams}.
\newblock
\urldef\tempurl%
\url{https://github.com/nexmark/nexmark}
\showURL{%
\tempurl}


\bibitem[Virtanen et~al\mbox{.}(2020)]%
        {scipy}
\bibfield{author}{\bibinfo{person}{Pauli Virtanen}, \bibinfo{person}{Ralf
  Gommers}, \bibinfo{person}{Travis~E. Oliphant}, \bibinfo{person}{Matt
  Haberland}, \bibinfo{person}{Tyler Reddy}, \bibinfo{person}{David
  Cournapeau}, \bibinfo{person}{Evgeni Burovski}, \bibinfo{person}{Pearu
  Peterson}, \bibinfo{person}{Warren Weckesser}, \bibinfo{person}{Jonathan
  Bright}, \bibinfo{person}{St{\'e}fan~J. {van der Walt}},
  \bibinfo{person}{Matthew Brett}, \bibinfo{person}{Joshua Wilson},
  \bibinfo{person}{K.~Jarrod Millman}, \bibinfo{person}{Nikolay Mayorov},
  \bibinfo{person}{Andrew R.~J. Nelson}, \bibinfo{person}{Eric Jones},
  \bibinfo{person}{Robert Kern}, \bibinfo{person}{Eric Larson},
  \bibinfo{person}{C~J Carey}, \bibinfo{person}{{\.I}lhan Polat},
  \bibinfo{person}{Yu Feng}, \bibinfo{person}{Eric~W. Moore},
  \bibinfo{person}{Jake {VanderPlas}}, \bibinfo{person}{Denis Laxalde},
  \bibinfo{person}{Josef Perktold}, \bibinfo{person}{Robert Cimrman},
  \bibinfo{person}{Ian Henriksen}, \bibinfo{person}{E.~A. Quintero},
  \bibinfo{person}{Charles~R. Harris}, \bibinfo{person}{Anne~M. Archibald},
  \bibinfo{person}{Ant{\^o}nio~H. Ribeiro}, \bibinfo{person}{Fabian Pedregosa},
  \bibinfo{person}{Paul {van Mulbregt}}, {and} \bibinfo{person}{{SciPy 1.0
  Contributors}}.} \bibinfo{year}{2020}\natexlab{}.
\newblock \showarticletitle{{{SciPy} 1.0: Fundamental Algorithms for Scientific
  Computing in Python}}.
\newblock \bibinfo{journal}{\emph{Nature Methods}}  \bibinfo{volume}{17}
  (\bibinfo{year}{2020}), \bibinfo{pages}{261--272}.
\newblock
\urldef\tempurl%
\url{https://doi.org/10.1038/s41592-019-0686-2}
\showDOI{\tempurl}


\bibitem[Wadler(1990)]%
        {wadler_deforest}
\bibfield{author}{\bibinfo{person}{Philip Wadler}.}
  \bibinfo{year}{1990}\natexlab{}.
\newblock \showarticletitle{Deforestation: transforming programs to eliminate
  trees}.
\newblock \bibinfo{journal}{\emph{Theoretical Computer Science}}
  \bibinfo{volume}{73}, \bibinfo{number}{2} (\bibinfo{year}{1990}),
  \bibinfo{pages}{231--248}.
\newblock
\showISSN{0304-3975}
\urldef\tempurl%
\url{https://doi.org/10.1016/0304-3975(90)90147-A}
\showDOI{\tempurl}


\bibitem[Wang et~al\mbox{.}(2017)]%
        {blaze}
\bibfield{author}{\bibinfo{person}{Xinyu Wang}, \bibinfo{person}{Isil Dillig},
  {and} \bibinfo{person}{Rishabh Singh}.} \bibinfo{year}{2017}\natexlab{}.
\newblock \showarticletitle{Program Synthesis Using Abstraction Refinement}.
\newblock \bibinfo{journal}{\emph{Proc. ACM Program. Lang.}}
  \bibinfo{volume}{2}, \bibinfo{number}{POPL}, Article \bibinfo{articleno}{63}
  (\bibinfo{date}{dec} \bibinfo{year}{2017}), \bibinfo{numpages}{30}~pages.
\newblock
\urldef\tempurl%
\url{https://doi.org/10.1145/3158151}
\showDOI{\tempurl}


\bibitem[Wang et~al\mbox{.}(2018)]%
        {relish-oopsla18}
\bibfield{author}{\bibinfo{person}{Yuepeng Wang}, \bibinfo{person}{Xinyu Wang},
  {and} \bibinfo{person}{Isil Dillig}.} \bibinfo{year}{2018}\natexlab{}.
\newblock \showarticletitle{Relational Program Synthesis}.
\newblock \bibinfo{journal}{\emph{Proc. ACM Program. Lang.}}
  \bibinfo{volume}{2}, \bibinfo{number}{OOPSLA}, Article
  \bibinfo{articleno}{155} (\bibinfo{date}{oct} \bibinfo{year}{2018}),
  \bibinfo{numpages}{27}~pages.
\newblock
\urldef\tempurl%
\url{https://doi.org/10.1145/3276525}
\showDOI{\tempurl}


\bibitem[Watson(1980)]%
        {polyinterp}
\bibfield{author}{\bibinfo{person}{G.~A. Watson}.}
  \bibinfo{year}{1980}\natexlab{}.
\newblock \bibinfo{booktitle}{\emph{Approximation theory and numerical
  methods}}.
\newblock \bibinfo{publisher}{John Wiley \& Sons}.
\newblock


\bibitem[Welford(1962)]%
        {welford}
\bibfield{author}{\bibinfo{person}{B.~P. Welford}.}
  \bibinfo{year}{1962}\natexlab{}.
\newblock \showarticletitle{Note on a Method for Calculating Corrected Sums of
  Squares and Products}.
\newblock \bibinfo{journal}{\emph{Technometrics}}  \bibinfo{volume}{4}
  (\bibinfo{year}{1962}), \bibinfo{pages}{419--420}.
\newblock
\urldef\tempurl%
\url{https://api.semanticscholar.org/CorpusID:120126049}
\showURL{%
\tempurl}


\bibitem[Yuan et~al\mbox{.}(2023)]%
        {syrup}
\bibfield{author}{\bibinfo{person}{Yongwei Yuan}, \bibinfo{person}{Arjun
  Radhakrishna}, {and} \bibinfo{person}{Roopsha Samanta}.}
  \bibinfo{year}{2023}\natexlab{}.
\newblock \showarticletitle{Trace-Guided Inductive Synthesis of Recursive
  Functional Programs}.
\newblock \bibinfo{journal}{\emph{Proc. ACM Program. Lang.}}
  \bibinfo{volume}{7}, \bibinfo{number}{PLDI}, Article \bibinfo{articleno}{141}
  (\bibinfo{date}{jun} \bibinfo{year}{2023}), \bibinfo{numpages}{24}~pages.
\newblock
\urldef\tempurl%
\url{https://doi.org/10.1145/3591255}
\showDOI{\tempurl}


\bibitem[Zaharia et~al\mbox{.}(2010)]%
        {spark}
\bibfield{author}{\bibinfo{person}{Matei Zaharia}, \bibinfo{person}{Mosharaf
  Chowdhury}, \bibinfo{person}{Michael~J. Franklin}, \bibinfo{person}{Scott
  Shenker}, {and} \bibinfo{person}{Ion Stoica}.}
  \bibinfo{year}{2010}\natexlab{}.
\newblock \showarticletitle{Spark: Cluster Computing with Working Sets}. In
  \bibinfo{booktitle}{\emph{Proceedings of the 2nd USENIX Conference on Hot
  Topics in Cloud Computing}} (Boston, MA)
  \emph{(\bibinfo{series}{HotCloud'10})}. \bibinfo{publisher}{USENIX
  Association}, \bibinfo{address}{USA}, \bibinfo{pages}{10}.
\newblock


\bibitem[Zaharia et~al\mbox{.}(2013)]%
        {spark-streaming}
\bibfield{author}{\bibinfo{person}{Matei Zaharia}, \bibinfo{person}{Tathagata
  Das}, \bibinfo{person}{Haoyuan Li}, \bibinfo{person}{Timothy Hunter},
  \bibinfo{person}{Scott Shenker}, {and} \bibinfo{person}{Ion Stoica}.}
  \bibinfo{year}{2013}\natexlab{}.
\newblock \showarticletitle{Discretized Streams: Fault-Tolerant Streaming
  Computation at Scale}. In \bibinfo{booktitle}{\emph{Proceedings of the
  Twenty-Fourth ACM Symposium on Operating Systems Principles}} (Farminton,
  Pennsylvania) \emph{(\bibinfo{series}{SOSP '13})}.
  \bibinfo{publisher}{Association for Computing Machinery},
  \bibinfo{address}{New York, NY, USA}, \bibinfo{pages}{423–438}.
\newblock
\showISBNx{9781450323888}
\urldef\tempurl%
\url{https://doi.org/10.1145/2517349.2522737}
\showDOI{\tempurl}


\bibitem[Zhang et~al\mbox{.}(2021)]%
        {clis}
\bibfield{author}{\bibinfo{person}{Guoqiang Zhang}, \bibinfo{person}{Yuanchao
  Xu}, \bibinfo{person}{Xipeng Shen}, {and} \bibinfo{person}{I\c{s}\i{}l
  Dillig}.} \bibinfo{year}{2021}\natexlab{}.
\newblock \showarticletitle{UDF to SQL Translation through Compositional Lazy
  Inductive Synthesis}.
\newblock \bibinfo{journal}{\emph{Proc. ACM Program. Lang.}}
  \bibinfo{volume}{5}, \bibinfo{number}{OOPSLA}, Article
  \bibinfo{articleno}{112} (\bibinfo{date}{oct} \bibinfo{year}{2021}),
  \bibinfo{numpages}{26}~pages.
\newblock
\urldef\tempurl%
\url{https://doi.org/10.1145/3485489}
\showDOI{\tempurl}


\end{thebibliography}

\newpage
\appendix
\section{Proofs}
\subsection{Proving Theorem \ref{thm:method-sound}} \label{proof:thm-method-sound}
In the following proof, we write $xs'$ to denote the stream corresponding to list $xs$.

\begin{proof}
We prove a slightly stronger theorem that if $\scheme \models \rfs$ then for all $xs$, $\denot{P}_{xs} = \last(\denot{\scheme}_{xs'})$ and that $\rfs(xs, \foldl(\prog', \init, xs))$ is true. The correctness of the proof relies on the fact that for all streams $xs'$, $\last(\denot{\scheme}_{xs'}) = \fst(\foldl(\prog', \init, xs))$. 
\begin{itemize}[leftmargin=*]
    \item {\bf Base Case}: $xs = []$.  Without loss of generality, let $\init = (c_1, \ldots, c_n)$. Since $xs = []$, we have that $\foldl(\prog', \init, []) = \init$ which means $\last(\denot{\scheme}_{xs'}) = c_1$. Furthermore, since $\scheme \models \rfs$ we can infer that  $\rfs[y_1]([]) = E([]) = \denot{P}_{[]} = c_1$ and that $\rfs([], \foldl(\prog', \init, []))$ is true. Thus, the theorem holds on the empty list.
    \item {\bf Inductive Case}: Assuming the theorem holds for list $xs$, consider a list of the form $ys = xs \pp [x]$ with corresponding stream $ys'$. Since $\last(\denot{\scheme}_{ys'}) = \fst(\foldl(\prog', \init, ys))$ we have from the definition of $\foldl$ that $\last(\denot{\scheme}_{ys'}) = \fst(\prog'( \foldl(\prog', \init, xs), x)$. From our inductive hypothesis, we know that $\rfs(xs, \foldl(\prog', \init, xs))$ is true and since $\scheme \models \rfs$ we know that $\rfs$ is inductive so we can infer that $\rfs(ys, \prog'(\foldl(\prog', \init, xs), x))$ is true. Thus,  $\rfs(ys, \foldl(\prog, \init, ys))$ holds which means $\rfs[y_1](ys) = \fst(\foldl(\prog, \init, ys)) = \last(\denot{\scheme}_{ys'})$. Finally, since $\rfs[y_1] = E$ we know that $\denot{P}_{ys} = E(ys) = \rfs[y_1](ys)$ which allows us to conclude that $\denot{P}_{ys} =  \last(\denot{\scheme}_{ys'})$. 
\end{itemize}
\end{proof}

\subsection{Proving Theorem \ref{thm:rfs_completeness}}
\begin{proof}
Let $\prog$ and $\scheme = (\init, \prog')$ be defined as in the theorem such that $\prog \simeq \scheme$.
Let $\lambda xs. \lambda v. \phi$ be the invariant for the expression $\foldl(\prog', \init, xs)$ such that $\phi \equiv \bigwedge_i v_i = E_i$ with $E_1 = E$.
We construct the RFS to be $\rfs$.
First, we show that condition (2) of our methodology holds. It follows that the invariant $\rfs([], \foldl(\prog', \init, [])) \equiv \rfs([], \init) \equiv \init \models \rfs$.

Then, we show that $\prog'$ is inductive relative to $\rfs$. Since the expression $\foldl(\prog', \init, xs)$ has an inductive invariant $\rfs$, we have the following Hoare triple holds:
\[
\hoare{\rfs(xs, {v})}{\quad {v}' := \foldl(\prog', \init, xs \pp [x]); \ xs' = xs \listappend [x] \quad \quad  } {\rfs(xs', {v}')},
\]
and by the axiom of $\foldl$, we derive the following
\[
\hoare{\rfs(xs, {v})}{\quad {v}' := \prog'\left( \foldl(\prog', \init, xs), x \right); \ xs' = xs \listappend [x] \quad \quad  } {\rfs(xs', {v}')}.
\]
Finally, by {the definition of $\foldl(\prog', \init, xs)$ inductive relative to $\rfs$}, we have $v = \foldl(\prog', \init, xs),$ which gives
\[
\hoare{\rfs(xs, {v})}{\quad {v}' := \prog'\left( v, x \right); \ xs' = xs \listappend [x] \quad \quad  } {\rfs(xs', {v}')}.
\]
Thus, $\prog'$ is inductive relative to $\rfs$. 
\end{proof}
\subsection{Proving Theorem \ref{thm:synthesizeexpr_soundness}}
\begin{proof}
We show that $E'$, the expression returned by {\sc SynthesizeExpr}, is equivalent to $E$ modulo $\rfs$, that is,
\[
\rfs(xs, y) \models E' = E[(xs\pp[x])/xs].
\]
\begin{itemize}
    \item [Case 1:] Let us assume that line 4 of the {\sc SynthesizeExpr} procedure returns an expression $E'$.
    Then $E'$ is deduced by matching $\chi$ with $\hole = E'$, where $\chi$ is obtained from the \textsc{FindImplicate} function called in line 2 of \Cref{algo:expr-synth}. Note that the implicate template passed to \textsc{FindImplicate} is $E[(xs\pp[x])/xs] = \hole$, and the invocation to \textsf{ElimQuantifier} on line 12 preserves the implicate formula.
    Additionally, the conjunction formula on line 10 specifies the condition (namely, RFS and axioms) under which the implicate template holds.
    Thus, by soundness of quantifier elimination procedure, we can infer that $E'$ is an expression satisfying $\rfs(xs, y) \models E' = E[(xs\pp[x])/xs]$.
    \item [Case 2:] Suppose that line 7 of {\sc SynthesizeExpr} returns $E'$.
    Then $E'$ is the result of the enumerative synthesis procedure with templatized expressions $\theta$.
    Similarly to Case 1, the soundness of quantifier elimination procedure implies the under-approximation of $E'$, each $t$ on line 18,
    is equivalent modulo $\rfs$ when the length of the input list is fixed, namely $k$ here.
    Then, we resort to the soundness of equivalence checking in {\sc EnumSynthesize} for any synthesized expression $E'$ such that $E' = E[(xs\pp[x])/xs]$.
\end{itemize}
In both cases, we conclude $E' = E[(xs\pp[x])/xs]$ and thus $E'$ is equivalent to $E$ modulo $\rfs$.
\end{proof}
\subsection{Proving Theorem \ref{thm:synthesize_soundness}}
\begin{lemma} \label{lem:decompose_soundness}
    Suppose {\sc Decompose$(\rfs, \prog)$} returns $(\prog', \context)$.
    If $\textsf{Holes}(\prog') = \left\{ h_1, \ldots, h_n \right\}$ can be completed by $E_1, \ldots, E_n$ such that $\rfs(xs, y) \models E_i = \context[i][(xs\pp[x])/xs]$,
    then the completed program obtained by $\prog' = \prog'[E_1/h_1, \ldots, E_n/h_n]$ satisfies $\rfs(xs, y) \models \fst(\prog'(y, x)) = \prog(xs')$.
\end{lemma}
\begin{proof}

We proceed the proof by structural induction over the decomposition rules presented in \Cref{fig:decomp}.

First, we show that every non-\textsc{Prog} rule preserves soundness. That is, given an offline expression $E$, if the decomposition rules returns the sketch $\sketch$ and its corresponding context $\context$, the completed sketch denoted as $E'$, obtained by replacing $\textsf{Holes}(\sketch)$ with $e_1, \ldots, e_n$ such that $\rfs(xs, y) \models e_i = \context[i][(xs\pp[x])/xs]$, is equivalent to $E$ modulo $\rfs$.
\begin{itemize}[leftmargin=*]
    \item \textsc{\textbf{Leaf.}} Suppose $E$ is a leaf expression. The \textsc{Leaf} rule returns $\sketch = E$. It is trivial to conclude as $\textsf{Holes} = \varnothing$.

    \item \textsc{\textbf{List.}} Suppose $E = L$ is a list expression. The \textsc{List} rule returns $\rfs \vdash L \rewrite \Box$ and $ \set{\Box \mapsto L}$. Let $E'$ be the online expression that completes $\sketch = \Box$.
    It follows our assumption that $\rfs(xs, y) \models E' = L[(xs\pp[x])/xs]$.

    \item \textsc{\textbf{Func.}} Suppose $E = g(E_1, \ldots, E_n)$. The \textsc{Func} rule returns $g(\sketch_1, \ldots, \sketch_n)$ and $\context_1 \cup \ldots \cup \context_n$. Let $E'_1, \ldots, E'_n$ be the completion of $\sketch_1, \ldots, \sketch_n$.
    By IH, we have $\rfs(xs, y) \models E'_1 = E_1[(xs\pp[x])/xs] \wedge \ldots \wedge E'_n = E_n[(xs\pp[x])/xs]$.
    Thus, given the completion of $\sketch$ as $g(E'_1, \ldots, E'_n)$, we have
\begin{align*}
    \rfs(xs, y) \models g(E'_1, \ldots, E'_n) &= g(E_1[(xs\pp[x])/xs], \ldots, E_n[(xs\pp[x])/xs]) \\
            &= g(E_1, \ldots, E_n)[(xs\pp[x])/xs] = E[(xs\pp[x])/xs].
\end{align*}

    \item \textsc{\textbf{ITE.}} Suppose $E = \ite{E_1}{E_2}{E_3}$ and $\rfs \vdash E_i \rewrite \sketch_i, \context_i$ for $E_1, E_2$ and $E_3$. Let $E'_1, E'_2$ and $E'_3$ be the completion of $\sketch_1$, $\sketch_2$ and $\sketch_3$ correspondingly, and $\sketch$ would be completed as $\ite{E'_1}{E'_2}{E'_3}$.
By IH, we have $\rfs \models E'_i = E_i$. Therefore,
\begin{align*}
    \rfs(xs, y) \models \ite{E'_1}{E'_2}{E'_3} &= \ite{E_1[(xs\pp[x])/xs]}{E_2[(xs\pp[x])/xs]}{E_3[(xs\pp[x])/xs]} \\
    &= (\ite{E_1}{E_2}{E_3})[(xs\pp[x])/xs] = E[(xs\pp[x])/xs].
\end{align*}
\end{itemize}

It remains to show the top-level decomposition rule, namely {\sc Prog}, preserves soundness.
Suppose the input to the {\sc Prog} rule is $P = \lambda xs. E$, $\dom(\rfs) = \set{y_1, \ldots, y_n}$.
Suppose $\rfs \vdash \rfs[y_i] \rewrite \sketch_i, \context_i $ for $i \in \set{1, \ldots, n}$ and the completion $\prog' = \lambda (y_1, \ldots, y_n). \lambda x. (E'_1, \ldots, E'_n)$.
Let $E'_i$ be the completion of $\sketch_i$ for $i \in \set{1, \ldots, n}$.
By IH, we have $\rfs(xs, y) \models E'_i = \rfs[y_i][(xs\pp[x])/xs]$. Now, since $y_i = E = P(xs)$ from line 2 of \Cref{algo:rfs}, we conclude
\begin{align*}
    \rfs(xs, y) \models \fst(\prog'(y, x)) = E'_1 = E[(xs\pp[x])/xs] = P(xs').
\end{align*}
\end{proof}
\begin{lemma} \label{lem:synthesizeonlineprog_equivmod}
    If {\sc SynthesizeOnlineProg$(\prog, \rfs)$} returns $\prog'$, then we have $\rfs(xs, y) \models \fst(\prog'(y, x)) = \prog(xs')$.
\end{lemma}
\begin{proof}
    It follows \Cref{thm:synthesizeexpr_soundness} that expression $E_i$ returned in line 4 satisfies
\begin{align}
    \label{eqn:lemma_3_equiv_mod}
    \rfs(xs, y) \models E_i = \context[i][(xs\pp[x])/xs].
\end{align}
    From line 5 and 7, we know that {\sc SynthesizeOnlineProg} would return $\prog'$ if and only if it finds all $E_i$ satisfying \Cref{eqn:lemma_3_equiv_mod}.
    By \Cref{lem:decompose_soundness}, the completed program obtained by $\prog' = \prog'[E_1/h_1, \ldots, E_n/h_n]$ in line 6 satisfies $\rfs(xs, y) \models \fst(\prog'(y, x)) = \prog(xs')$.
\end{proof}
\begin{lemma} \label{lem:synthesizeonlineprog_model}
    If {\sc SynthesizeOnlineProg$(\prog, \rfs)$} returns $\prog'$ and $\init$ is a model of $\rfs[xs \mapsto \nil]$,
    then the online scheme $\scheme = (\init, \prog') \models \rfs$.
\end{lemma}

\begin{proof}
    To show that $\scheme = (\init, \prog') \models \rfs$, we need to show both (1) $\init \models \rfs$ and (2) $\prog'$ is inductive relative to $\rfs$. (1) is easily obtained by soundness of \textsf{Model} invoked on line 3 of \Cref{algo:toplevel} that finds the model of $\rfs[xs \mapsto \nil]$. We proceed to (2).
    
    Suppose
\begin{align*}
    \rfs(xs, y) &= \bigwedge_{i=1}^n y_i = \rfs[y_i] \\
                &= y_1 = \prog(xs) \wedge y_2 = E_2(xs) \wedge \ldots \wedge y_n = E_n(xs)
\end{align*}
holds and let
\begin{align*}
    \rfs(xs', y') &= \bigwedge_{i=1}^n y'_i = \rfs[y'_i] \\
                  &= y'_1 = \prog(xs') \wedge y'_2 = E_2(xs') \wedge \ldots \wedge y'_n = E_n(xs').
\end{align*}
We show $\rfs(xs', y')$ holds where $xs' := xs\pp[x]$ and $y' := \prog'(y, x)$.

By \Cref{lem:synthesizeonlineprog_equivmod}, we have
\[
\rfs(xs, y) \models \fst(y') = \fst(\prog'(y, x)) = \prog(xs\pp[x]) = \prog(xs') = y'_1.
\]
Meanwhile, by \Cref{thm:synthesizeexpr_soundness}, we have $\rfs(xs, y) \models E'_i = E_i[xs\pp[x]/xs]$ for all $E'_i$ that {\sc SynthesizeExpr} returns making up $\prog'$. This implies that for every offline expression $E_i$ and its corresponding online expression $E'_i$,
\[
\rfs(xs, y) \models y'_i = E'_i = E_i(xs\pp[x]/xs) = E_i(xs'),
\]
which concludes (2) that
\[
\hoare{\rfs(xs, {y})}{\quad {y}' := \prog'(y, x); \ xs' = xs \listappend [x] \quad \quad  } {\rfs(xs', {y}')}.
\]
\end{proof}

Finally, we proceed to the proof of \Cref{thm:synthesize_soundness}.
\begin{proof}
From line 2-5 of \Cref{algo:toplevel}, {\sc Synthesize$(\prog)$} returns $\scheme = (\init, \prog')$ if and only if (1) it finds an $\init$ such that $\init$ is a model of $\rfs[xs \mapsto \nil]$ and (2) {\sc SynthesizeOnlineProg$(\prog, \rfs)$} returns $\prog'$.
By \Cref{lem:synthesizeonlineprog_model}, $\scheme = (\init, \prog') \models \rfs$.

We conclude that $\prog \simeq \scheme$ by \Cref{thm:method-sound}.
\end{proof}

\section{Solving Templates via Polynomial Interpolation} \label{sec:polyinterp}

In this section, we describe an optimization where we use polynomial interpolation to synthesize online expressions from the templates generated by {\sc MineExpressions}. We observed in many cases that if the online program uses an auxiliary parameter $n$ denoting the number of stream elements processed, then the templates produced by {\sc MineExpressions} can be transformed into equivalent online expressions by replacing the unknowns with polynomials over $n$. Intuitively, the reason this case occurs frequently when the online procedure needs to use $n$ is because {\sc MineExpression} generates the template by unrolling the input list $k$ times where $k$ is a constant. As such, the underlying QE engine is implicitly aware that $n = k$ so it often uses $k$, instead of $n$, to generate the expression; this results in an expression with several constants specific to the unrolling depth when in fact those constants should be replaced with expressions over $n$.

Our top level procedure procedure for solving templates via polynomial interpolation is presented in Algorithm \ref{algo:poly}. {\sc SolveTemplate} takes as input an offline expression $E$, an RFS $\rfs$, along with a template $T$, and returns an equivalent online expression $E'$ or $\bot$ indicating the procedure failed. Without loss of generality, we assume $T$ is of the form $\sum_{i=1}^m{??_ie_i}$ where $e_i$ is an expression over auxiliary variables. The algorithm assumes that the desired online expression can be synthesized by replacing each unknown $??_i$ with a univariate polynomial $P_i$ over $n$. As such, it first sets the candidate online expression $E'$ to $T$ (line 2) and then iteratively transforms $E'$ by replacing unknowns with synthesized polynomials.  Next, it calls $\textsc{SamplePoints}$, which returns a map $\mathcal{M}$ mapping each unknown $??_i$ to a set of $k$ distinct points $\mathcal{P}_i$ on $P_i$ (line 3). If $k$ points cannot be sampled, {\sc SolveTemplate} returns $\bot$ indicating the template could not be solved. For each unknown, {\sc SolveTemplate} uses polynomial interpolation to generate the polynomial $P_i$ (line 7) of lowest degree over $n$ which passes through all $k$ points. It then updates $E'$ by replacing $??_i$ with $P_i$ (line 8). After replacing all the unknowns with polyomials, {\sc SamplePoints} returns $E'$ if it is equivalent to $E$ modulo $\rfs$  and $\bot$ otherwise (lines 9-11).

\begin{figure}[!t]
\begin{algorithm}[H]
\caption{Template completion algorithm via polynomial interpolation}
\label{algo:poly}
\begin{algorithmic}[1]
\Procedure{\textsc{SolveTemplate}}{$E$, $\rfs$, $T$}
\vspace{2pt}
\Statex \textbf{Input:} An offline expression $E$
\Statex \textbf{Input:} An rfs $\rfs$
\Statex \textbf{Input: } A template $T := \sum_{i=1}^n{??_ie_i}$
\Statex \textbf{Output: } An online program $E'$ equivalent to $E$ modulo $\rfs$ or $\bot$
\State $E' \gets T$
\State $\mathcal{M} \gets \textsc{SamplePoints}(E, \rfs, T, $k$)$
\If{$\mathcal{M} = \bot$}
\State\Return $\bot$
\EndIf
\For{$(??_i, \mathcal{P}_i) \in \mathcal{M}$}
\State $P_i \gets \textsf{Interpolate}(\mathcal{P}_i)$
\State $E' \gets E'[??_i/P_i]$
\EndFor
\If{$\rfs \models E = E'$}
\State\Return $E'$
\EndIf
\State\Return $\bot$
\vspace{2pt}
\EndProcedure
\end{algorithmic}
\end{algorithm}
\end{figure}

\begin{figure}[!t]
\begin{algorithm}[H]
\caption{Samples points on the polynomials for each unknown}
\label{algo:samplepts}
\begin{algorithmic}[1]
\Procedure{\textsc{SamplePoints}}{$E$, $\rfs$, $T$, $k$}
\vspace{2pt}
\Statex \textbf{Input:} An offline expression $E$
\Statex \textbf{Input:} An rfs $\rfs$
\Statex \textbf{Input: } A template $T := \sum_{i=1}^n{??_ie_i}$
\Statex \textbf{Input: } A constant $k$ denoting the number of samples that should be collected for each $??_i$.
\Statex \textbf{Output: } A map $\mathcal{M}$ mapping each unknown $??_i$ to a set of $k$ points sampled on its corresponding polynomial or $\bot$ if points cannot be sampled.
\State $\mathcal{L} \gets \textsf{SampleLengths}(k)$
\State $m \gets \textsf{NumUnknowns}(T)$
\State $\mathcal{M} \gets \{\}$
\For{$l \in L$}
\State $X \gets \textsf{SampleLists}(m, l)$
\State $\Lambda \gets \textsf{GenEquations}(E, T, X, \rfs)$
\State $\psi \gets \textsf{Solve}(\Lambda)$
\If{$\psi = \bot$}
\State\Return $\bot$
\EndIf
\For{$??_i \in \textsf{Unknowns}(T)$}
\State $\mathcal{M}[??_i].\textsf{add}((l, \psi(??_i)))$
\EndFor
\EndFor
\State\Return $\mathcal{M}$
\vspace{2pt}
\EndProcedure
\end{algorithmic}
\end{algorithm}
\end{figure}

\bfpara{SamplePoints}. The heart of the above procedure is $\textsc{SamplePoints}$ (presented in Algorithm \ref{algo:samplepts}), which samples points on each $P_i$ \emph{even without knowing $P_i$}.  The key insight behind $\textsc{SamplePoints}$ is the following. Suppose the desired online expression $E'$ is equivalent to $E$ modulo RFS $\rfs$ and $E'$ can be derived from $T$ by replacing unknowns in $T$ with univariate polynomials $P_i$ over $n$. Then the following must be true: for any lists $xs$ and $xs'$ of constant length $k$, 
\begin{align}
E(xs) &= \sum_{i=1}^n{P_i(k)f_i(xs)} \\
E(xs') &= \sum_{i=1}^n{P_i(k)f_i(xs')}
\end{align}
where $f_i = e_i[v_i/\rfs[v_i]. \  \forall v_i \in \textsf{Vars}(e_i)]$.
Since $P_i$ only depends on $n$, its value will not change for different lists of the same length. Thus, given a template $T$ with $m$ unknowns, we can recover $P_i(k)$ by sampling $m$ distinct lists $xs_1, \dots, xs_m$ of length $k$ and solving the following system of linear equations:

\begin{equation}\label{eq:system}
\begin{bmatrix}
f_1(xs_1) & f_2(xs_1) & \dots & f_n(xs_1) \\
f_1(xs_2) & f_2(xs_2) & \dots & f_n(xs_2) \\
\dots  & \dots  & \dots  & \dots \\
f_1(xs_m) & f_2(xs_m) & \dots & f_n(xs_m)
\end{bmatrix}
\begin{bmatrix}
??_1 \\ ??_2 \\ \dots \\ ??_m 
\end{bmatrix}
=
\begin{bmatrix}
E(xs_1) \\ E(xs_2) \\ \dots \\ E(xs_m)
\end{bmatrix}
\end{equation}

With this background explained, we can now describe Algorithm \ref{algo:samplepts}. It starts by sampling $k$ distinct lengths (line 2) and for each length $l$ performs the following steps:

\begin{itemize}
    \item Sample $m$ distinct lists of length $l$ (line 6)
    \item Generate a system of equations like the one presented in Equation \ref{eq:system} (line 7)
    \item Solve the system of equations to obtain $P_i(l)$ or return $\bot$ if the equations have no solution (8-10)
    \item Update the map $\mathcal{M}$ to add the point $(l, P_i(l))$ for each unknown (lines 11-12).
\end{itemize}

\tool chooses $k$ to be 11 which means the polynomials generated by this procedure can be at most degree 11. In our experience, we have not encountered a case where the desired polynomial was of degree larger than four.

\end{document}